\documentclass[sigconf]{acmart}

\newif\ifpublish
\publishtrue

\AtBeginDocument{}

\usepackage{amsfonts}
\usepackage{algorithm}
\usepackage[noend]{algpseudocode}
\usepackage{graphicx}
\usepackage{textcomp}
\usepackage{xcolor}
\usepackage{xspace}
\usepackage{cleveref}
\usepackage{pifont}
\usepackage{IEEEtrantools}
\usepackage{tikz}

\usetikzlibrary{shapes.geometric, arrows, positioning, intersections    }

\newcommand{\codelink}{
    \ifpublish
        \url{https://github.com/asonnino/mysticeti/tree/stingray}
        (commit \texttt{0ae4bb5})
    \else
        Code available but link omitted for blind review.
    \fi
}

\newcommand{\para}[1]{\vskip 0.5em\noindent\textbf{#1.}}

\algnewcommand{\LineComment}[1]{\State {} \(\triangleright\) \textit{#1}}
\algnewcommand{\InlineRequire}[1]{\State \textbf{require} {#1}}

\Crefname{figure}{Fig.}{Figs.}
\Crefname{figure}{Fig.}{Figs.}
\Crefname{table}{Tab.}{Tabs.}
\Crefname{table}{Tab.}{Tabs.}
\Crefname{section}{Sec.}{Secs.}
\Crefname{section}{Sec.}{Secs.}
\Crefname{appendix}{App.}{Apps.}
\Crefname{appendix}{App.}{Apps.}

\Crefname{algorithm}{Alg.}{Algs.}
\Crefname{algorithm}{Alg.}{Algs.}
\Crefname{line}{l.}{ll.}
\Crefname{line}{l.}{ll.}

\Crefname{proposition}{Prop.}{Props.}
\Crefname{proposition}{Prop.}{Props.}
\Crefname{lemma}{Lem.}{Lems.}
\Crefname{lemma}{Lem.}{Lems.}
\Crefname{theorem}{Thm.}{Thms.}
\Crefname{theorem}{Thm.}{Thms.}
\Crefname{corollary}{Cor.}{Cors.}
\Crefname{corollary}{Cor.}{Cors.}
\Crefname{definition}{Def.}{Defs.}
\Crefname{definition}{Def.}{Defs.}
\Crefname{item}{Item}{Items}
\Crefname{item}{Item}{Items}

\newcommand{\sysname}{Stingray\xspace}
\newcommand{\sys}{\sysname}
\newcommand{\fun}{FastUnlock\xspace}
\newcommand{\bcobject}{bounded counter\xspace}

\newcommand{\BCObject}{Bounded Counter\xspace}
\newcommand{\bcobjects}{bounded counters\xspace}

\newcommand{\Globalsafety}{Global safety\xspace}

\newcommand{\globalsafetybc}{global safety\xspace}
\newcommand{\Globalsafetybc}{Global safety\xspace}

\newcommand{\convergence}{eventual consistency\xspace}
\newcommand{\Convergence}{Eventual consistency\xspace}

\newcommand{\freal}{f_{\mathrm{r}}}

\newcommand{\none}{\text{\textbf{None}}\xspace}
\newcommand{\unlocked}{\text{\textbf{Unlocked}}\xspace}
\newcommand{\confirmed}{\text{\textbf{Confirmed}}\xspace}
\newcommand{\noop}{\text{\textbf{No-Op}}\xspace}

\newcommand{\Cert}{\textsf{Cert}\xspace}
\newcommand{\Oid}{\textsf{ObjectId}\xspace}
\newcommand{\Okey}{\textsf{ObjectKey}\xspace}
\newcommand{\Okeys}{\textsf{ObjectKeys}\xspace}
\newcommand{\Version}{\textsf{Version}\xspace}

\newcommand{\Bal}{\textsf{Bal}\xspace}
\newcommand{\Val}[1]{\Delta(#1)\xspace}
\newcommand{\initBal}{\textsf{Bal}_0\xspace}

\newcommand{\Bud}{\textsf{Bud}\xspace}
\newcommand{\BudTotal}{\overline{\textsf{Bud}}\xspace}
\newcommand{\budFrac}{\eta}

\newcommand{\certtxs}{\mathcal{C}\xspace}
\newcommand{\senttxs}{\textsf{sentTxs}\xspace}
\newcommand{\diff}{\delta}
\newcommand{\req}{v\xspace}
\newcommand{\PrevTxs}{\textsf{PrevTxs}\xspace}
\newcommand{\PrevVersion}{\textsf{PrevVersion}\xspace}
\newcommand{\PrevVersions}{\textsf{PrevVersions}\xspace}

\newcommand{\SignedTxs}{\textsf{SignedTxs}\xspace}
\newcommand{\initVersion}{v_0}
\newcommand{\parent}[1]{p({#1})}
\newcommand{\parents}[1]{P({#1})}

\newcommand{\history}[1]{H_{#1}}
\newcommand{\inctxs}[1]{\history{#1}\xspace}
\newcommand{\exctxs}[1]{\widetilde{H}_{#1}\xspace}
\newcommand{\excnegtxs}[1]{{\widetilde{H}_{#1}}^{-}\xspace}
\newcommand{\excpostxs}[1]{{\widetilde{H}_{#1}}^{+}\xspace}

\newcommand{\tx}{\textsf{Tx}\xspace}
\newcommand{\cert}{\textsf{Cert}\xspace}
\newcommand{\auth}{\textsf{Auth}\xspace}
\newcommand{\unlockreq}{\textsf{UnlockRqt}\xspace}
\newcommand{\unlockvote}{\textsf{UnlockVote}\xspace}
\newcommand{\unlockcert}{\textsf{UnlockCert}\xspace}
\newcommand{\effectvote}{\textsf{EffectSign}\xspace}
\newcommand{\effectcert}{\textsf{EffectCert}\xspace}

\newcommand{\funcCall}[2]{\emph{#1}(#2)\xspace}

\newcommand{\valid}[1]{\emph{valid}(#1)\xspace}

\newcommand{\sign}[1]{\emph{sign}(#1)\xspace}

\newcommand{\exec}[1]{\emph{exec}(#1)\xspace}

\newcommand{\allVersionsMatch}[1]{\emph{allVersionsMatch}(#1)\xspace}
\newcommand{\areCertified}[1]{\emph{areCertified}(#1)\xspace}
\newcommand{\broadcast}[1]{\emph{broadcast}(#1)\xspace}

\newcommand{\lockdb}{\textsc{LockDb}\xspace}
\newcommand{\unlockdb}{\textsc{UnlockDb}\xspace}

\newcommand{\one}{\ding{202}\xspace}
\newcommand{\two}{\ding{203}\xspace}
\newcommand{\three}{\ding{204}\xspace}
\newcommand{\four}{\ding{205}\xspace}
\newcommand{\five}{\ding{206}\xspace}
\newcommand{\six}{\ding{207}\xspace}
\newcommand{\seven}{\ding{208}\xspace}
\newcommand{\eight}{\ding{209}\xspace}

\newcommand{\CO}{{\mathcal{O}}}

\newcommand{\CS}{{\mathcal{S}}}

\newcommand{\CV}{{\mathcal{V}}}

\definecolor{myParula01Blue}{RGB}{0,114,189}
\definecolor{myParula02Orange}{RGB}{217,83,25}
\definecolor{myParula03Yellow}{RGB}{237,177,32}
\definecolor{myParula04Purple}{RGB}{126,47,142}
\definecolor{myParula05Green}{RGB}{119,172,48}
\definecolor{myParula06LightBlue}{RGB}{77,190,238}
\definecolor{myParula07Red}{RGB}{162,20,47}

\usepackage{pgfplots}
\pgfplotsset{compat=newest}

\pgfplotsset{
    mysimpleplot/.style = {
        every axis plot/.prefix style={thick},
        width=\linewidth,
        height=0.45\linewidth,
        xlabel={Client-Perceived Throughput (tx/s)},
        ylabel={Latency (s)},
        xmin=1, xmax=2e5,
        ymin=0, ymax=4,
        title style={font=\scriptsize,align=center},
        legend cell align=left,
        legend style={font=\scriptsize},
        legend columns=2,
        legend style={
                at={(0.5,1)},
                yshift=0.3em,
                anchor=south,
                draw=none,
                /tikz/every even column/.append style={
                        column sep=0.3em
                    },
                cells={
                        align=left
                    }
            },
        grid=both,
        minor tick num=4,
        major grid style={solid,very thin,draw=gray!50},
        minor grid style={solid,ultra thin,draw=gray!20},
        label style={font=\scriptsize,align=center},
        tick label style={font=\scriptsize},
        yticklabel style={
            /pgf/number format/fixed,
            /pgf/number format/precision=2
        },
        scaled y ticks=false,
    },
    line1/.style = {
        myParula02Orange, mark=triangle, mark size=2pt, mark options={solid}, dotted, line width=1pt, 
        error bars/.cd, y dir=plus,y explicit,
    },
    line2/.style = {
        myParula02Orange, mark=*, mark size=2pt, mark options={solid}, solid, line width=1pt, 
        error bars/.cd, y dir=plus,y explicit,
    },
    line3/.style = {
        myParula05Green, mark=square, mark size=2pt, mark options={solid}, dashed, line width=1pt, 
        error bars/.cd, y dir=plus,y explicit,
    },
    line4/.style = {
        myParula05Green, mark=diamond*, mark size=2pt, mark options={solid}, solid, line width=1pt, 
        error bars/.cd, y dir=plus,y explicit,
    },
}

\begin{document}

\title{\sysname: Fast Concurrent Transactions Without Consensus}

\ifpublish
  \author{Srivatsan Sridhar}
  \affiliation{
    \institution{Stanford University}
    \city{}
    \state{}
    \country{}
  }

  \author{Alberto Sonnino}
  \affiliation{
    \institution{Mysten Labs \\ University College London}
    \city{}
    \country{}
  }

  \author{Lefteris Kokoris-Kogias}
  \affiliation{
    \institution{Mysten Labs}
    \city{}
    \country{}
  }

  \renewcommand{\shortauthors}{Sridhar, Sonnino, Kokoris-Kogias}
\else
  \author{}
\fi

\setcopyright{none}
\renewcommand\footnotetextcopyrightpermission[1]{}
\settopmatter{printacmref=false}
\settopmatter{printfolios=true}

\begin{abstract}
  
Recent advances have improved the throughput and latency of blockchains
by processing transactions accessing different parts of the state concurrently.
However, these systems
are unable
to concurrently process (a) transactions accessing the same state, even if they are (almost) commutative, e.g., payments much smaller than an account’s balance,
and
(b) multi-party transactions, e.g., asset swaps.
Moreover, they are slow to recover
from contention, requiring once-in-a-day synchronization.
We present \sysname, a novel blockchain architecture that addresses these limitations.
The key conceptual contributions are
a replicated bounded counter that processes (almost) commutative transactions concurrently,
and a \fun protocol that uses a fallback consensus protocol for fast contention recovery.
We prove \sysname's security in an asynchronous network with Byzantine faults and demonstrate on a global testbed that \sysname achieves 10,000 times the throughput of prior systems for commutative workloads.

\end{abstract}

\maketitle
\pagestyle{plain}

\section{Introduction}

Blockchain technology has fundamentally transformed the landscape of digital transactions,
providing a decentralized and secure framework for managing digital assets. Despite its
innovative potential, current blockchains~\cite{bitcoin, ethereum, tendermint} face significant scalability and
efficiency challenges. Two key issues are the batch commit of blocks, which necessitates
transactions to wait until a block is ready to be proposed, and the sequential execution of
transactions, which fails to exploit modern CPU architectures capable of parallel processing.
These limitations introduce latency and restrict throughput, impeding the full utilization of
available resources.

To mitigate these bottlenecks, recent advances have focused on enhancing concurrency in
transaction processing. This has been achieved through improvements in both the agreement module
via consensus-less fast paths~\cite{sui, fastpay, abc, astro} and the execution module via parallel
execution engines~\cite{blockstm, pilotfish, sui,solana}. Existing consensus-less fast paths enable concurrency and significantly reduce latency for
transactions that have a consensus number of 1~\cite{guerraoui19consensus}, such as payments and transfers, and facilitate easy identification of transactions accessing independent resources for parallel execution~\cite{sui,solana}.

Despite these advancements, existing fast-path protocols and parallel execution engines still
have limitations.
First, while they parallelize transactions touching independent parts of the state, e.g., payments to/from different accounts, transactions accessing the same account are sequential
because
they incorporate a version number derived from the preceding transaction on
the same account.
This lack of concurrency limits performance gains
for commutative transactions that
do not inherently require sequentiality.

Second,
the fast path has been confined to a limited range of transactions. For example,
in systems such as Sui~\cite{suiL}, FastPay~\cite{fastpay} and Astro~\cite{astro}, fast-path transactions are restricted to objects owned by a single entity, allowing
simple payments and asset transfers but excluding more complex operations such as asset swaps or
multi-signature authorizations. This limitation significantly decreases the potential
transaction load that could benefit from the fast path.

Finally, if a user submits concurrent conflicting transactions through the fast path, the system must
lock the object until the system can deterministically resolve this deadlock. While it is
anticipated that users will not send conflicting transactions, such conflicts may arise due to
minor bugs in wallet implementations or malicious collusion among signatories in multi-signature
structures. Deployed systems, such as Sui~\cite{sui}, can take up to 24 hours to resolve these conflicts,
resulting in a suboptimal user experience.

We address the first limitation by
allowing concurrent transactions that are commutative or almost commutative.
We were inspired by the database literature, where systems can deduce (or be instructed) that
accessing the same resource can occur safely when actions are commutative~\cite{kleppmann2021thinking}.
This is captured using Conflict-free Replicated Data Types (CRDTs)~\cite{nasirifard2023orderlesschain}.
However, using CRDTs is insufficient, as any transaction in a blockchain system must pay for gas. This gas payment (and any payment transaction) has a non-commutative comparison with zero. To our knowledge, the only algorithms that bridge this gap between commutativity and comparing a counter (or set) with a bound (i.e., nearly commutative)
have been proposed under crash faults~\cite{nasirifard2023orderlesschain,weidner2022oblivious,almeida2013scalable} and break under Byzantine faults.
We resolve this open question in \sysname with \emph{the first Byzantine fault-tolerant bounded counter}. The intuition behind this construction is to (a) provide a local budget per validator for signing transactions but (b) require quorums for approval such that even if every quorum includes all malicious validators the global budget can still not be overspent. This allows users to spend up to half their account's value concurrently before resynchronization.
Importantly, resynchronization still happens through a consensus-less fast path, solving the problem in asynchrony using protocols resembling reliable broadcast~\cite{cachin2011introduction}.

Once we identify that the real problem in concurrent distributed ledgers is not the concurrent accesses on the same memory location (i.e., \emph{concurrency}), but the non-monotonic~\cite{keep-calm-crdt-on} accesses that cannot be reordered (i.e., \emph{contention}),
we address the second limitation by extending the fast path to accommodate transactions involving objects owned by
different users, such as asset swaps, and support complex authorization structures, including
threshold and logical combinations of user authorizations. These transactions are finalized extremely fast and
in parallel as long as there is no contention over the resources accessed. However, the absence of contention is no longer guaranteed if one of the users is malicious,
risking loss of liveness.
The last challenge we solve is alleviating this risk through our novel \fun protocol that leverages the consensus of hybrid blockchains~\cite{suiL,mysticeti} to swiftly resolve conflicting transactions and
enable users to submit new transactions within the latency of consensus protocols (e.g., 400ms for Mysticeti~\cite{mysticeti}).

We implement \sysname on top of the consensus-less fast path of Sui~\cite{sui}, called Mysticeti-FPC~\cite{mysticeti}, and show that it provides significant throughput and latency improvements for transaction loads concurrently accessing the same resources with commutative operations (20,000 tps in 0.5 seconds instead of 2 tps). Additionally, we demonstrated that \sysname incurs no performance penalty for fully parallelizable workloads, both with and without faults, when compared to Sui.

\para{Contributions}
This paper makes the following contributions:
\begin{itemize}
    \item We introduce the first bounded counter with Byzantine fault tolerance, a new object type that supports commutative operations, such as addition and subtraction, and some non-commutative operations, such as comparison with zero for account balances, without consensus.
    \item We propose a new protocol, \fun, that enables the fast resolution of conflicting transactions within the latency of consensus protocols.
    \item We present \sysname, a novel system that applies these techniques to existing fast-path protocols. In addition to bounded counters and \fun, \sysname supports on the fast path, transactions involving objects owned by different users, such as asset swaps, and complex authorization structures, including thresholds and other logical combinations.
    \item We formally prove the safety and liveness of \sysname in the asynchronous network model with Byzantine faults.
    \item We implement \sysname and evaluate it in a realistic geo-replicated environment to demonstrate that it outperforms the state of the art by $10{,}000$x for a commutative workload.
\end{itemize}

\section{Background} \label{sec:background}

Several consensus-less systems have been proposed in the literature, including FastPay~\cite{fastpay}, Astro~\cite{astro}, Zef~\cite{zef}, and Linera~\cite{linera}. In this section, we recap Sui (the Sui Lutris mechanism~\cite{sui}) as a basis for the \sysname design, as it is the only currently deployed system supporting a consensus-less fast path.
\sysname improves upon the Sui design (a) the concurrency admitted between transactions by exploiting commutativity
and (b) the scope of object authentication and transactions in the fast path. Finally, \sysname recognizes that
these benefits apply to the fast path only in the optimistic case
and mitigates this by (c) extending Sui's consensus with \fun to mitigate the worst case of contented transactions.

\subsection{Data Structures}
\label{sec:background-data-structures}

\para{Object Types}
The Sui blockchain state consists of a set of objects categorized into three types. Sui determines whether to use the fast path or the consensus path for a transaction based on the types of objects involved.
\begin{itemize}
    \item \emph{Read-only objects} cannot be mutated or deleted and may be used in transactions in either path concurrently by all users.
    \item \emph{Owned objects} have an owner field that determines access control. The owner is an address representing a public key. A transaction may access the object if it is signed by that key (which can also be a multi-signature). 
    A canonical example is a user's cryptocurrency account.
    As owned objects are never under contention when the owner is honest, Sui validates transactions accessing only owned or read-only objects using the fast path.
    \item \emph{Shared objects} do not specify an owner. They can instead be included in transactions by anyone and do not require any authorization. Instead, their authorization logic is enforced by a smart contract. In Sui, such objects are only accessed through consensus to serialize their access.
\end{itemize}

In \sysname, we introduce two additional types of objects, \bcobjects (\Cref{sec:bounded-counter}) and collective objects (\Cref{sec:collective-objects}). Transactions using them are validated using the fast path.

\tikzset{
    process/.style={rectangle, draw, rounded corners, minimum height=2em, minimum width=3cm, text centered, thick},
    circleprocess/.style={circle, draw, fill=white, thick, minimum size=0.25cm, text centered},
    arrow/.style={thick,->,>=stealth},
}

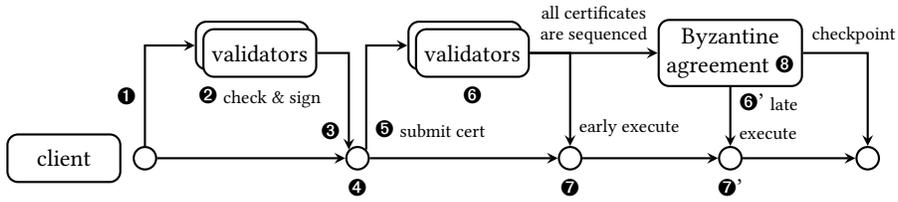
\begin{figure*}[t]
    \centering
    \begin{tikzpicture}[]
        \node (client) [process, minimum width=1.5cm] {client};
        \node (submittx) [circleprocess, right=0.15cm of client] {};
        \node (validators1b) [process, right=0.5cm of submittx, yshift=1.5cm, minimum width=1.5cm] {}; 
        \node (validators1f) [process, right=0.1cm of validators1b.west, yshift=-0.1cm, fill=white, minimum width=1.5cm] {validators}; 
        \draw [arrow] (submittx.north) |- (validators1b.west) node [anchor=east, pos=0.25] {\one}; 
        \node [below=0cm of validators1f] {\two \footnotesize{check \& sign}};
        \node (submitcert) [circleprocess, right=2.5cm of submittx] {};
        \node (validators2b) [process, right=0.5cm of submitcert, yshift=1.5cm, minimum width=1.5cm] {}; 
        \node (validators2f) [process, right=0.1cm of validators2b.west, yshift=-0.1cm, fill=white, minimum width=1.5cm] {validators}; 
        \draw [arrow] (submittx.east) -- (submitcert.west);
        \node [below=0cm of submitcert] {\four};
        \draw [arrow] (validators1f.east) -| (submitcert.north west) node [pos=0.9, anchor=east] {\three};
        \draw [arrow] (submitcert.north east) |- (validators2b.west) node [anchor=west, pos=0.1, align=left] {\five \footnotesize{submit cert}}; 
        \node [below=0cm of validators2f] {\six};
        \node (earlyexecute) [circleprocess, right=2.5cm of submitcert] {};
        \node (byzagreement) [process, right=1cm of earlyexecute, yshift=1.4cm, minimum width=1.5cm, align=center] {Byzantine\\agreement \eight};
        \draw [arrow] (validators2f.east) -| (earlyexecute.north) node [pos=0.9, anchor=west, font=\footnotesize] {early execute};
        \draw [arrow] (validators2f.east) -- (byzagreement.west) node [midway, anchor=south, font=\footnotesize, align=center] {all certificates\\are sequenced};
        \draw [arrow] (submitcert.east) -- (earlyexecute.west);
        \node [below=0cm of earlyexecute] {\seven};
        \node (lateexecute) [circleprocess, below=1.4cm of byzagreement.center, anchor=center] {};
        \draw [arrow] (earlyexecute.east) -- (lateexecute.west);
        \draw [arrow] (byzagreement.south) -- (lateexecute.north) node [midway, anchor=west, align=left] {\six' \footnotesize{late}\\\footnotesize{execute}};
        \node [below=0cm of lateexecute] {\seven'};
        \node [circleprocess] (checkpoint) [right=1.5cm of lateexecute] {};
        \draw [arrow] (lateexecute.east) -- (checkpoint.west);
        \draw [arrow] (byzagreement.east) -| (checkpoint.north) node [pos=0, anchor=south west, font=\footnotesize] {checkpoint}; 
    \end{tikzpicture}
    \caption{General protocol flow of Sui Lutris~\cite{suiL} fast-path (\one-\seven) \& consensus failover system (\eight,\six',\seven').}
    \label{fig:overview}
\end{figure*}

\para{Transactions} A transaction is a signed command that specifies several input objects, a version number per object, and a set of parameters.
For owned objects, executing the transaction
consumes the input object versions and constructs a set of output objects---which may be the input objects at a later version or new objects.
Shared objects do not require a specified version. Instead, the system assigns the version on which the transaction executes based on the consensus sequence.
Input objects' versions must be the latest version in validators' databases and must not be re-used across transactions.
This limits concurrency because validators must process transactions with the same object sequentially, and we address this using \bcobject objects (\Cref{sec:bounded-counter}).

In Sui, a transaction is signed by a single address and therefore can use one or more objects owned by that address. A single transaction cannot use objects owned by more than one address and must use shared objects instead.
In this work, we will allow a transaction to use objects owned by different addresses if the transaction is signed by the all the owners (\Cref{sec:multi-owner-txs}), thus enabling validation of such transactions on the fast path.

\para{Certificates}
A \emph{certificate} ($\cert$) on a transaction contains the transaction and signatures from a quorum of at least $2f+1$ validators with their identifiers.
A certificate may not be unique, and the same logical certificate may be signed by different quorums of validators. However, two different valid certificates on the same transaction are treated as representing semantically the same certificate.

\subsection{Processing in the Fast Path and Consensus}
\Cref{fig:overview} provides an overview of Sui and, by extension, \sysname's common case. A transaction is sent by a user to all validators~(\one), who ensure it is correctly authenticated by the owners of all owned objects
and that all objects exist~(\two).
A correct validator rejects any conflicting transaction using the same owned object versions (the first transaction using an object acquires a \emph{lock} on it).
Validators then countersign the transaction~(\three) and return the signatures to the user. A quorum of signatures constitutes a \emph{certificate} for the transaction~(\four). Anyone may submit the certificate to the validators~(\five) that check it.

At this point, execution may take the fast path: If the certificate only references read-only and owned objects
(in \sysname, also owned \bcobject and collective objects)
it is executed immediately~(\six) and a signature on the effects of the execution is returned to the user. Signatures from $2f+1$ validators create an \emph{effects certificate}~(\seven), and the transaction is \emph{finalized}, i.e., it is guaranteed to never be rolled back, even if the set of validators change. If any shared objects are included, execution must wait for them to be assigned versions post-consensus. In all cases, certificates are input into consensus and sequenced~(\eight). Once sequenced, the system assigns a common version number to shared objects for each certificate, and execution can resume (steps~\six' and~\seven') to finalize the transaction. The common sequence of certificates is also used to construct checkpoints, which are guaranteed to include all finalized transactions~(\eight).

\para{Checkpoints and Reconfiguration} \label{sec:sui-checkpoints}
Sui ensures transaction finality before consensus for owned object transactions~(\seven) or after consensus for shared object transactions~(\seven').
A reconfiguration protocol ensures that any transaction finalized through the fast path will eventually be included in a checkpoint before the end of the epoch (epochs last roughly for 24 hours).

\para{Limitations of Sui}
Sui users are not allowed to submit conflicting transactions that reuse the same owned object versions in steps~(\one) and~(\two), limiting concurrency.
If a misconfigured user client behaves this way, neither transaction may successfully construct a certificate~(\four), and the owned object becomes locked until the end of the epoch, harming user experience. To avoid mistrusting users from locking each other's objects for a day, Sui restricts transactions to only contain objects from a single owner, thus limiting the applicability of the fast path. For similar reasons, objects used in the fast path may have at most one owner.

\sysname addresses all these limitations through changes to the fast path (steps \one-\seven).
\sysname increases the concurrency in Sui by allowing concurrent transactions on \bcobjects.
This can be trivially extended to bounded sets and any commutative objects (e.g., add-only counters or PN-sets~\cite{shapiro2011conflict}).
We also allow transactions with multi-owner authentication and collective objects that are not expected to have contention to execute in the fast path.
Since the risk of contention in such transactions is low but not zero, we finally show how to recover from contention using \fun.
\section{System Overview} \label{sec:overview}
We introduce \sysname and the setting in which it operates.

\subsection{Threat Model and Goals} \label{sec:model}
The adversary is computationally bounded, ensuring that standard cryptographic properties such as the security of hash functions, digital signatures, and other primitives hold. Under this assumption \sysname ensures \emph{validity}, meaning that all transactions that are executed have valid authorization (\Cref{def:valid-transaction}).

We consider a message-passing system with $n = 3f + 1$ validators running the \sysname protocol. An adversary can adaptively corrupt up to $f$ validators, referred to as \emph{Byzantine}, who may deviate arbitrarily from the protocol. The remaining validators, called \emph{honest}, follow the protocol.
The communication network is asynchronous and messages can be delayed arbitrarily.
Given these conditions, \sysname is \emph{safe}, that is, the union (merge) of all
transactions executed by honest validators
does not lead to invalid state transitions (\Cref{def:safety}).
This safety property subsumes the classic safety of blockchains since any fork in the state, if merged, would violate validity predicates such as conservation of value (e.g., in case of double-spends).
At the same time, this definition captures concurrent execution.

Finally, assuming messages among honest validators are eventually delivered, \sysname is \emph{live}, meaning honest validators eventually execute certifiably valid user transactions and update their state accordingly, and in the absence of new transactions, eventually converge to the same state (\Cref{def:liveness}).
Here, liveness is required for transactions that can be executed without causing invalid state transitions, as certified by a quorum of validators (\Cref{def:certificate-validity}).
In contrast, requiring liveness for all transactions issued by honest users is too strong and unachievable, since the validity of their state transitions depends on the system's current state.

\begin{definition}[Transaction Validity]
    \label{def:valid-transaction}
    A transaction $\tx$ is valid if
    all its input objects are owned by the transaction's signers.
\end{definition}

\begin{definition}[Certificate Validity]
    \label{def:certificate-validity}
    A transaction $\tx$ has a valid certificate if:
    \begin{itemize}
        \item The transaction is valid
        \item It has a quorum certificate signed by at least $2f+1$ validators
    \end{itemize}
\end{definition}

Let $T_p(t)$ denote the set of transactions executed by validator $p$ up to time $t$.
The applications specify a set of predicates  $\mathcal{P}$ over sequences of transactions.
For example, a predicate may require that no user can spend more than their account balance
or that no two transactions perform conflicting state updates.

\begin{definition}[Safety Properties]
    \label{def:safety}
    For any execution of \sysname with at most $f$ Byzantine validators:
    \begin{itemize}
        \item \textbf{Validity:} For all $p,t$: all $\tx \in T_p(t)$
              have valid certifcates.

        \item \textbf{\Globalsafety:} For all $t$, for all subsets $H$ of honest validators, there exists a sequence $T$ that contains $\bigcup_{p \in H} T_p(t)$ such that for all $P \in \mathcal{P}$: $P(T)$ is true.
    \end{itemize}
\end{definition}

\begin{definition}[Liveness Properties]
    \label{def:liveness}
    For any execution of \sysname with at most $f$ Byzantine validators:
    \begin{itemize}
        \item \textbf{Progress:} Every
              transaction with a valid certificate
              is eventually executed by all honest validators,
              unless an owner of its input objects equivocates, i.e., signs two transactions with the same object version as input.

        \item \textbf{\Convergence:} For all honest validators $p_1, p_2$ and all time $t$, there exists a time $t' \geq t$ such that
              $T_{p_1}(t') \supseteq T_{p_2}(t)$.

    \end{itemize}
\end{definition}

Alongside the above properties, our goal is to enable concurrent execution of as many transactions as possible.
\sysname uses separate paths for processing commutative/bounded-counter transactions and for other fast-path transactions. We prove the security of the former in \Cref{sec:bounded-counter-owned-proof} and the latter in \Cref{sec:proofs}.

\subsection{Motivating Applications of \sys}\label{sec:applications}

We present three example applications of \sysname: one utilizing the bounded counter and two employing multi-owner transactions (whose potential liveness risks are mitigated by \fun).

\para{Concurrent Payments}
The first application, concurrent debit or credit of an account, is one of the most common in existing blockchain systems.
Not only is this useful for payments, but it is also useful for gas debits, which are required for every transaction.
However, it presents significant challenges for parallel execution if not carefully designed.
Incoming transactions (credits) only increase the account's balance, so they are commutative.
Since they do not cause contention, they can easily be processed concurrently (even though prior blockchains do not do so).
Concurrent debits, on the other hand, are more complex due to the need for zero-balance checks.
These operations are non-commutative, hence they cause true contention that hinders concurrency.

Concurrent transactions may be achieved in UTXO-style blockchains by dividing one's account into smaller UTXOs. However, keeping track of these smaller UTXOs is cumbersome, and using the same UTXO twice is actually a double-spending attack.
In account- or object-based blockchains, one can similarly split their account into smaller accounts and concurrently access them, but it has the same challenges as with UTXOs.
For example, users on Sui frequently make mistakes, causing contention, which results in losing access for an entire day.
\sys resolves this issue without introducing such complications by using a mostly-commutative bounded counter, allowing half of the budget to be spent concurrently before requiring a sequential rebalancing transaction.

\para{Atomic swaps}
Atomic swaps enable two parties to exchange digital assets without relying on a trusted intermediary. While consensus-based blockchains achieve this through smart contracts, consensus-less environments face the risk of deadlock due to Byzantine users issuing concurrent transactions. Such scenarios can effectively lock both parties’ assets. Thus, in Sui, swaps require multiple transactions, with at least one (the swap) relying on consensus.

However, the risk of liveness loss on the fast path only occurs when an active attacker deliberately creates contention. \sys provides an effective mitigation of this risk with the \fun protocol. This safety net allows \sysname to support multi-owner transactions in the fast path, allowing fast path atomic swaps and other multi-party smart contracts, enhancing the programmability of consensus-less transactions.
Although not a real application, the same risk runs for users who inadvertently equivocate on their objects, leading to unexpected deadlocks and a poor user experience. Here too, \fun helps restore liveness quickly, lowering barriers to securely using an ultra-low latency blockchain.

\para{Regulated stablecoins} Regulated stablecoins~\cite{circle} require the issuer to be able to
block an account for regulatory reasons, besides its owner spending
from it.
This has eluded consensus-less systems since sequencing these potentially conflicting operations requires consensus.
Yet, the ability to
block objects is nearly never exercised, creating no practical contention. \sysname's collective
objects, that may be used by more than one owner (or complex
access control) enables such transactions in the fast path.

\subsection{Challenges}
\sysname defines new object types that allow for higher concurrency. We do not focus on purely commutative data structures for which a Byzantine fault-tolerant CRDT~\cite{kleppmann2021thinking} is sufficient but still cannot support transactions for blockchains because of gas payments that require
a comparison with zero, a non-commutative operation.

To resolve this, we define the \emph{bounded counter}, which runs in the fast path.
The design of this bounded counter creates our first challenge (\textbf{Challenge~1}): implementing nearly commutative objects in BFT settings.
To achieve this, we depart from prior work in CFT~\cite{almeida2013scalable} that splits the budget among replicas since in the BFT setting, a single malicious replica could sign infinite transactions.
As a result, we need to rely on quorums to distribute the budget collectively. Since there is an exponential number of potential quorums we reduce this to giving a sufficient budget to each validator to spend concurrently but not enough that an overspend can happen if all quorums have a minority of equivocating participants that spend infinite amounts. This tension results in our construction spending half of the bound when not under attack. We then show how we can reset the budget through a consistent read, allowing for the full amount to be concurrently spent with only $\log n$ points of synchronization.

The second challenge we take on is that existing consensus-less blockchains require transactions to operate on state owned by a single user, due to the fear of concurrency.
\sysname absolves concurrency for consensus-less transactions and instead identifies contention as the true culprit of correctness violations.
Thus, \sysname enhances programmability by defining new types of objects such as collective objects that are owned by multiple users,
as well as new types of transactions, such as multi-owner transactions (e.g., an asset swap) that are processed in the fast path.
This approach is powerful, and promises reduced latency for such operations that currently require consensus. But this creates our second challenge (\textbf{Challenge~2}): users may naturally submit conflicting transactions because contention is now possible. For example, two users may perform a swap and end up locking the objects due to bad timing. To address this challenge, \sysname uses a novel unlocking mechanism, called \fun, that allows users to resolve conflicts quickly, enabling users to submit new transactions within the latency of a consensus protocol.

\section{Concurrency through Commutativiy}
\label{sec:bounded-counter}

We present protocols that allow users to finalize a common category of transactions: updates to a \emph{\bcobject}, without using consensus.
This type of object is strictly harder to implement than CRDTs as it has a non-monotonic comparison with zero~\cite{keep-calm-crdt-on}.
Our \bcobject protocol is a modification to the fast-path protocol shown in \Cref{fig:overview}, while transactions on other owned objects continue to be processed as in \Cref{fig:overview}.

\subsection{The \BCObject Object} \label{sec:bounded-counter-object}

A bounded counter is an object that has a balance $\Bal \in \mathbb{R}$ as its state and supports
transactions with a parameter $\diff \in \mathbb{R}$,
allowing additions ($\delta > 0$) and subtractions ($\delta < 0$) on its balance, while maintaining the invariant $\Bal \geq 0$.
Our goal is to execute transactions concurrently to the extent possible.

\begin{definition}[Bounded counter] \label{def:bounded-counter}
    A bounded counter object has
    an authorized user called the owner
    and an initial balance $\initBal$,
    supports transactions $\tx$ with value $\tx.\delta \in \mathbb{R}$
    and has the following properties:
    \begin{enumerate}

        \item \textbf{\Globalsafetybc}: as in \Cref{def:safety} with predicate $P(T) = (\initBal + \sum_{\tx \in T} \tx.\delta \geq 0)$.

        \item \textbf{Progress}: If an honest owner sends a set of transactions $T$ such that $\initBal + \sum_{\tx \in T} \tx.\delta \geq 0$, then all validators will eventually
              execute all transactions in $T$.
        \item \textbf{Validity} and \textbf{\convergence}: as in \Cref{def:safety,def:liveness}.
    \end{enumerate}
\end{definition}

The progress property is a strengthening of \Cref{def:liveness} wherein we specify the conditions under which transactions issued by honest owners get certified and eventually executed.
We show a protocol in \Cref{sec:bounded-counter-owned} and prove that it achieves the above properties in \Cref{sec:bounded-counter-owned-proof}.

The bounded counter is useful for a common use case of blockchains, cryptocurrencies, where an account's balance can be represented using a \bcobject.
The bounded counter can be generalized to a bounded set in which elements can be added or removed, as long as the size of the set does not exceed a predefined bound.
This may be useful, for example, to mint a predefined number of limited-edition non-fungible tokens (NFTs).
More generally, the \bcobject can be used as a loop counter.
For simplicity, we describe the most common case of a real-valued counter bounded below by $0$ and consider payments from an account as the canonical use case.

\subsection{Key Ideas for the Bounded Counter} \label{sec:bounded-counter-key-ideas}

Consider a user (owner) who owns an account containing $\initBal$ units of money.
As an example, suppose that each transaction the user makes credits $1$ unit from its account, i.e., $\tx.\delta = -1$.
In this case,
safety requires that no honest validator execute more than $\initBal$ transactions.
Validators execute a transaction only if they see a certificate for the transaction, so it is sufficient to ensure that no more than $\initBal$ transactions get certified.

\para{Key Idea 1: Signing Budgets for Validators}
Recall that Sui's fast path prevents the certification of two transactions on the same version of an object by ensuring that validators sign at most one transaction per version (\Cref{sec:background}).
Extending this approach to allow concurrent transactions, we must ensure that each validator signs only a few transactions concurrently.
We assign to each validator a \emph{budget} $\Bud = \budFrac \initBal$, where $\budFrac = \frac{f+1}{2f+1}$, which is the maximum number of transactions the validator can sign.
Since each certified transaction is signed by at least $2f+1$ validators, $f+1$ of whom are honest,
for each certified transaction, at least $f+1$ is deducted from the total budget of all honest validators. Since the total budget of all honest validators starts at $(2f+1)\Bud$, the number of certified transactions can be at most $\frac{(2f+1)\Bud}{f+1} = \initBal$, even if the user and up to $f$ validators are malicious.

Thus, the signing budgets ensure \globalsafetybc.
However, this idea alone does not satisfy liveness. If the $f$ Byzantine validators abstain from signing transactions, each certified transaction requires signatures from $2f+1$ (all) \emph{honest} validators, causing every honest validator to decrease their budget. So, at most $\Bud = \eta \initBal$ transactions will get certified,
while liveness requires all of them to be certified eventually.

\tikzset{
    startstop/.style={rectangle, draw, dashed, text centered, minimum height=1.5em, minimum width=3em, font=\footnotesize},
    process/.style={circle, draw, minimum size=1.5em, text centered, font=\scriptsize},
    arrow/.style={thick,->,>=stealth},
    dashedarrow/.style={thick, dashed,->,>=stealth},
}

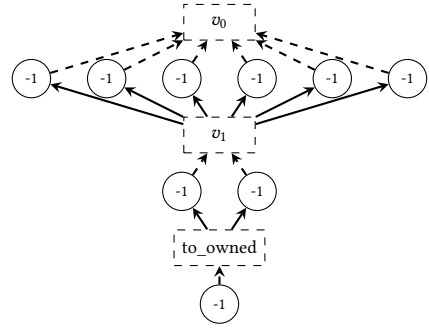
\begin{figure}
    \centering
    \begin{tikzpicture}[x=1cm,y=0.75cm]
        \node (v0) [startstop] at (0,0) {$\initVersion$};
        \node (node1) [process] at (0.5,-1) {-1};
        \node (node2) [process] at (1.5,-1) {-1};
        \node (node3) [process] at (2.5,-1) {-1};
        \node (node4) [process] at (-0.5,-1) {-1};
        \node (node5) [process] at (-1.5,-1) {-1};
        \node (node6) [process] at (-2.5,-1) {-1};
        \node (v1) [startstop] at (0,-2) {$v_1$};
        \node (node7) [process] at (0.5,-3) {-1};
        \node (node8) [process] at (-0.5,-3) {-1};
        \node (to_owned) [startstop] at (0,-4) {to\_owned};
        \node (node9) [process] at (0,-5) {-1};

        \draw [dashedarrow] (node1) -- (v0);
        \draw [dashedarrow] (node2) -- (v0);
        \draw [dashedarrow] (node3) -- (v0);
        \draw [dashedarrow] (node4) -- (v0);
        \draw [dashedarrow] (node5) -- (v0);
        \draw [dashedarrow] (node6) -- (v0);
        \draw [arrow] (v1) -- (node1);
        \draw [arrow] (v1) -- (node2);
        \draw [arrow] (v1) -- (node3);
        \draw [arrow] (v1) -- (node4);
        \draw [arrow] (v1) -- (node5);
        \draw [arrow] (v1) -- (node6);
        \draw [dashedarrow] (node7) -- (v1);
        \draw [dashedarrow] (node8) -- (v1);
        \draw [arrow] (to_owned) -- (node7);
        \draw [arrow] (to_owned) -- (node8);
        \draw [dashedarrow] (node9) -- (to_owned);
    
    \end{tikzpicture}
    \caption{Version updates in the bounded counter. Circles represent unit decrement transactions and dashed boxes versions. The initial balance is $\initBal = 9$, and $f = 1$. For version $\initVersion$, each validator has a budget of $\frac{f+1}{2f+1}\initBal = 6$. After $6$ transactions are certified, the user sends a version update (dashed box $v_1$) containing pointers to the $6$ certified transactions. The remaining balance is $9 - 6 = 3$ and the validators update their budget to $\frac{f+1}{2f+1}*3 = 2$. Finally, when the remaining balance is $1$, the user converts the bounded counter to a standard owned object and spends the remaining balance.
    }
    \label{fig:bc-owned-versions-example}
\end{figure}

\para{Key Idea 2: Version Updates}
When the user realizes that validators may have exhausted their budget (because $\Bud$ transactions have already been certified), the user sends a \emph{version update request} which includes pointers to all the previously certified transactions.
Upon seeing a valid version update request, each validator updates their budget
to $\budFrac$ fraction of the balance remaining after executing the certified transactions.
In the above example, this increases each validator's budget from $0$ to $\budFrac(1-\budFrac)\initBal$.
The validator also updates its local version to stop signing transactions with the previous version and start signing transactions with the new version (see \Cref{fig:bc-owned-versions-example} for an example).

Note that within a version, transactions get certified concurrently, while transactions across versions are certified sequentially.
At each version, the user can spend up to $\budFrac$ fraction of the remaining balance, until finally, when the remaining balance is small enough, the user can convert the bounded counter to a standard owned object and spend the remaining amount in one last transaction (shown in \Cref{fig:bc-owned-versions-example}).
Thus, with an initial balance of $\initBal$ and each transaction spending $1$ unit, the user requires only $\log(\initBal)$ version update requests, and therefore $O(\log(\initBal))$ latency to spend the entire balance.
In contrast, this would require $O(\initBal)$ latency on Sui's fast path or any other consensus protocol.

\subsection{Bounded Counter for Single Owner} \label{sec:bounded-counter-owned}
\Cref{alg:bounded-counter-owned} shows the algorithm run by validators for the bounded counter.
For simplicity, we specify the algorithm for when all transactions have only one bounded counter object as input, although transactions with multiple bounded counters and other owned objects as inputs can be processed in the fast path.

\para{Processing Transactions}
For each bounded counter object, the state maintained by each validator includes the current version $\Version$, the current budget $\Bud$, and the set of transactions $\SignedTxs$ the validator has signed (\Cref{loc:bc-owned-init-version,loc:bc-owned-init-budget,loc:bc-owned-init-signed}).
Upon receiving a new valid transaction, the validator checks that the transaction has the same version as the validator's current version (\Cref{loc:bc-owned-tx-version}).
If the transaction's value $\tx.\diff$ is within budget (\Cref{loc:bc-owned-tx-check-budget}), the validator deducts its budget if required, marks the transaction as signed, and sends the signature to the user.
Addition or debit transactions ($\tx.\diff > 0$) will always be signed, without decreasing the budget, while subtraction or credit transactions ($\tx.\diff < 0$) will be signed only up to the budget.

\para{Finalizing and Executing Transactions}
Upon receiving a valid certificate $\Cert$, the validator broadcasts the certificate to other validators, executes the transaction to update the counter's state locally, and sends an \emph{effects signature} to the user.
Unlike standard owned objects where the effects signature contains the new state of the object, for a bounded counter object, the effect signature simply marks the transaction as executed because the validator does not know the counter's correct state yet due to its partial view of other certified transactions.
Once the user receives $2f+1$ effects certificates, it considers the transaction finalized.

\para{Version Update Requests}
When the validator receives a version update request, it processes the request only if all the transactions to which the request points ($\req.\PrevTxs$) are certified (\Cref{loc:bc-owned-req-valid}).
Moreover, all the transactions in $\req.\PrevTxs$ must have the same version as the validator's current version (\Cref{loc:bc-owned-req-version}).
If so, the validator updates its budget.
For each certified transaction in $\req.\PrevTxs$, the validator increases for additions or decreases for subtractions its budget by $\budFrac$ times the transaction's value (\Cref{loc:bc-owned-req-update-budget}).
The validator also reclaims its spent budget for transactions it signed and are now included in $\req.\PrevTxs$ (\Cref{loc:bc-owned-req-regain-budget}), ensuring that budget is not deducted twice.
After the update, the validator's budget is $\budFrac$ times the remaining balance of the bounded counter after executing all certified transactions minus the values of the transactions that the validator signed but were not included in $\req.\PrevTxs$.
Finally, the validator updates its current version to $\req$ (\Cref{loc:bc-owned-req-update-version}).%
\footnote{In reality, the version identifier stored by the validator and included in each transaction can be a hash of the version update request. Collision resistance will bind each version identifier to a unique version update request.}

\begin{algorithm}[t]
    \caption{Bounded counter for single owner (validator logic)}
    \label{alg:bounded-counter-owned}
    \footnotesize
    \begin{algorithmic}[1]
        \LineComment{Initialize a bounded counter object}
        \LineComment{Initial version is $\initVersion$ and budget is computed from initial balance}
        \Procedure{InitBC}{$\initBal$}
        \State $\Version \gets \initVersion$ \Comment{current version for which the validator signs txs} \label{loc:bc-owned-init-version}
        \State $\Bud \gets \budFrac * \initBal$ \Comment{remaining budget for signing txs for $\Version$} \label{loc:bc-owned-init-budget}
        \State $\SignedTxs \gets \emptyset$ \Comment{set of txs signed so far} \label{loc:bc-owned-init-signed}
        \EndProcedure
        \vspace{1em}
        \LineComment{Executed upon receiving a transaction}
        \Procedure{ProcessTx}{$\tx$}
        \LineComment{Ensure same transaction is not proceesed twice}
        \If{$\tx \in \SignedTxs$} \Return $\sign{\tx}$ \EndIf \label{loc:bc-owned-already-signed}
        \InlineRequire{$\valid{\tx}$} \Comment{check $\tx$ has valid signatures} \label{loc:bc-owned-tx-valid}
        \InlineRequire{$\tx.\Version = \Version$} \label{loc:bc-owned-tx-version}
        \InlineRequire{$\Bud + \tx.\diff \geq 0$} \label{loc:bc-owned-tx-check-budget}
        \LineComment{Decrease the budget if decrement transaction}
        \If{$\tx.\diff < 0$} $\Bud \gets \Bud + \tx.\delta$  \label{loc:bc-owned-tx-decrease-budget}
        \EndIf
        \State $\SignedTxs \gets \SignedTxs \cup \{\tx\}$ \label{loc:bc-owned-tx-append-signed}
        \State \Return $\sign{\tx}$ \label{loc:bc-owned-tx-sign}
        \EndProcedure
        \vspace{1em}
        \LineComment{Upon receiving valid certificate signed by $2f+1$ validators; ensure same certificate not processed twice}
        \Procedure{ProcessCert}{$\cert$}
        \State $\broadcast{\cert}$ \label{loc:bc-owned-cert-broadcast}
        \State \textbf{wait until} $\funcCall{areExecuted}{\cert.\tx.\Version.\PrevTxs}$ \label{loc:bc-owned-cert-version}
        \If{$\valid{\cert}$} $\exec{\cert.\tx}$
        \Comment{execute tx, persist object's state} \label{loc:bc-owned-cert-execute}
        \EndIf
        \State \Return $\sign{\cert}$ \Comment{send signature to finalize $\cert.\tx$ to user} \label{loc:bc-owned-cert-effects-sign}
        \EndProcedure
        \vspace{1em}
        \LineComment{Upon receiving version update request signed by the user}
        \LineComment{Ensure same request is not processed twice}
        \Procedure{ProcessVersionUpdateReq}{$\req$}
        \State $\broadcast{\req}$ \label{loc:bc-owned-req-broadcast}
        \InlineRequire{$\req.\PrevVersion = \Version$
            \textbf{and} $\areCertified{\req.\PrevTxs}$} \label{loc:bc-owned-req-valid}
        \LineComment{Check all txs in $\req.\PrevTxs$ have version $\PrevVersion$}
        \InlineRequire{$\allVersionsMatch{\req.\PrevTxs, \req.\PrevVersion}$} \label{loc:bc-owned-req-version}
        \LineComment{Update budget based on txs in $\req.\PrevTxs$}
        \For{$\tx$ \textbf{in} $\req.\PrevTxs$}
        \State $\Bud \gets \Bud + \budFrac * \tx.\diff$ \label{loc:bc-owned-req-update-budget}
        \LineComment{Regain spent budget for txs included in $\req.\PrevTxs$}
        \If{$\tx \in \SignedTxs$ \textbf{and} $\tx.\diff < 0$} $\Bud \gets \Bud - \tx.\diff$ \EndIf \label{loc:bc-owned-req-regain-budget}
        \EndFor
        \LineComment{Start signing txs for the updated version}
        \State $\Version \gets \req$ \label{loc:bc-owned-req-update-version}
        \EndProcedure

    \end{algorithmic}

\end{algorithm}

\para{Honest Users}
The bounded counter protocol allows a user to get transactions certified concurrently under low contention, i.e., when the number of concurrent transactions is low. In \Cref{alg:bounded-counter-owned-user},
we specify how an honest user interacts with the bounded counter. In summary, the honest user ensures the following:
\begin{enumerate}
    \item The user sends transactions with the same $\Version$ until it sends a version update request (\Cref{alg:bounded-counter-owned-user}~\Cref{loc:bc-owned-user-create-tx})
    \item For any version, the user does not send transactions exceeding the validators' budget for that version (\Cref{alg:bounded-counter-owned-user}~\Cref{loc:bc-owned-user-check-budget}).
    \item Upon exhausting the budget for a version (\Cref{alg:bounded-counter-owned-user}~\Cref{loc:bc-owned-user-call-version-update}), the user sends a version update request containing pointers to all transactions it sent for that version (\Cref{alg:bounded-counter-owned-user}~\Cref{loc:bc-owned-user-create-req}).
    \item When the validators' budget falls to $0$ for a certain version (perhaps due to rounding), the user sends a transaction requesting to convert the bounded counter object to a standard owned object so that the small amount of remaining balance may be spent through a classic owned object transaction.
\end{enumerate}
An honest user can satisfy these requirements by keeping track of the current version and the validator's budget, all of which can be computed based on the transactions it has sent (see \Cref{alg:bounded-counter-owned-user}), without any communication.
In \Cref{sec:app-bcounter-algs-proofs}, we show that liveness holds for users that behave as specified by \Cref{alg:bounded-counter-owned-user}.

\begin{algorithm}[t]
    \caption{Bounded counter for single owner (user logic)}
    \label{alg:bounded-counter-owned-user}
    \footnotesize
    \begin{algorithmic}[1]
        \Procedure{Init}{$\initBal$}
        \State $\Version \gets \initVersion$
        \State $\Bud \gets \budFrac * \initBal$
        \State $\senttxs \gets \emptyset$
        \EndProcedure
        \vspace{1em}
        \LineComment{Invoked when application requires updating the bounded counter}
        \Procedure{Update}{$\diff$}
        \If{$\Bud + \diff < 0$} \Comment{If update exceeds the budget}
        \State $\Call{VersionUpdate}{ }$ \label{loc:bc-owned-user-call-version-update}
        \EndIf
        \LineComment{Abort if update exceeds the budget even after version update}
        \If{$\Bud + \diff < 0$} \Return Error \EndIf \label{loc:bc-owned-user-check-budget}
        \State $\tx \gets \{\Version \colon \Version, \diff \colon \diff \}$ \label{loc:bc-owned-user-create-tx}
        \If{$\diff < 0$} $\Bud \gets \Bud + \tx.\diff$ \EndIf \label{loc:bc-owned-user-tx-deduct-budget}
        \State $\funcCall{sendToValidators}{\tx}$
        \State $\senttxs \gets \senttxs \cup \{\tx\}$
        \EndProcedure
        \vspace{1em}
        \Procedure{VersionUpdate}{ }
        \State $\req \gets \{\PrevVersion: \Version, \PrevTxs: \senttxs\}$ \label{loc:bc-owned-user-create-req}
        \For{$\tx$ \textbf{in} $\senttxs$}
        \State $\Bud \gets \Bud + \budFrac * \tx.\diff$ \label{loc:bc-owned-user-req-update-budget}
        \If{$\tx.\diff < 0$} $\Bud \gets \Bud - \tx.\diff$ \EndIf \label{loc:bc-owned-user-req-regain-budget}
        \EndFor
        \LineComment{Send a version update request if there is enough remaining budget, otherwise request to convert the bounded counter to an owned object}
        \If{$\Bud \geq \textsf{minBudget}$} $\funcCall{sendToValidators}{\req}$
        \Else\xspace $\funcCall{sendToValidators}{\funcCall{convertToOwnedObject}{\senttxs}}$
        \EndIf
        \State $\Version \gets \req$
        \EndProcedure
    \end{algorithmic}
\end{algorithm}

\subsection{Security Proof}
\label{sec:bounded-counter-owned-proof}

\begin{theorem}
    \label{thm:bc-owned-validity}
    \Cref{alg:bounded-counter-owned} satisfies validity.
\end{theorem}
\begin{proof}
    Validators execute only valid transactions with a valid certificate (\Cref{alg:bounded-counter-owned}~\Cref{loc:bc-owned-cert-execute}).
\end{proof}

\begin{theorem}
    \label{thm:bc-owned-eventual-cons}
    \Cref{alg:bounded-counter-owned} satisfies eventual consistency.
\end{theorem}
\begin{proof}
    If one honest validator executes a transaction (\Cref{alg:bounded-counter-owned}~\Cref{loc:bc-owned-cert-execute}), the validator also broadcasts the certificate for the transaction (\Cref{loc:bc-owned-cert-broadcast}).
    Eventually, all validators also receive the version update request $\cert.\tx.\Version$ (along with their certified transactions) corresponding to the certified transaction because at least one honest validator (who signed the certified transaction) received the request and broadcast it (\Cref{loc:bc-owned-req-broadcast}).
    Thus, eventually, all honest validators will execute the transaction.
\end{proof}

As a warmup for proving \globalsafetybc, we first show how Key Idea 1 ensures that within a single version,
any subset of the certified transactions does not spend too much balance.
Refer to \Cref{tab:notation}
for a summary of the notation used in the proof.

\begin{table}[ht]
    \centering
    \begin{tabular}{|c|p{0.78\columnwidth}|}
        \hline
        $f$ & maximum number of adversarial validators \\
        $\initBal$ & initial balance of \bcobject \\
        $\Bud$ & validator's signing budget \\
        $\BudTotal_{v}$ & average budget of honest validators for version $v$ \\
        $\budFrac \triangleq \frac{f+1}{2f+1}$ & fraction of $\initBal$ assigned to $\Bud$ \\
        $\tx.\diff$ & quantity added by $\tx$ to the \bcobject \\
        $\Val{\CS}$ & sum of the quantities of transactions in the set $\CS$ \\
        $\initVersion$ & initial version of the \bcobject \\
        $\parent{v}$ & parent version of $v$, same as $v.\PrevVersion$ \\
        $\CV$ & set of versions with at least one certified transaction \\
        $\certtxs_{v}$ & set of certified transactions with version $v$ \\ 
        $\history{v}$ & history of $v$: transactions included in the version update requests of $v$ and all previous versions \\
        $\exctxs{v}$ & certified transactions with $v$ and all previous versions, except those in $\history{v}$ \\
        $\CS^{+}$, $\CS^{-}$ & Increment, decrement transactions in the set $\CS$ \\
        \hline
    \end{tabular}
    \caption{Table of notation}
    \label{tab:notation}
\end{table}

\begin{definition}
    The value function is defined on a set of transactions $\CS$ as $\Val{\CS} = \sum_{\tx \in \CS} \tx.\diff$.
\end{definition}

\begin{lemma}
    \label{lem:bc-owned-warmup}
    For any version $v$, let $\BudTotal_v$ be the average budget of all honest validators at the time when they set $\Version$ to $v$.%
    \footnote{Count a validator's budget as $0$ if it never sets $\Version$ to $v$.}
    Let $\certtxs_v$ be the set of certified transactions with version $v$.
    Then, for all $T \subseteq \certtxs_{v}$:
    $\Val{T} \geq -\frac{1}{\budFrac} \BudTotal_v$.
\end{lemma}
\begin{proof}

    Let $\certtxs_v^{-} \subseteq \certtxs_v$ be the set of decrement transactions ($\tx$ such that $\tx.\diff < 0$). It is sufficient to show that
    $\Val{\certtxs_v^{-}} \geq -\frac{1}{\budFrac} \BudTotal_v$.

    Suppose that $\freal \leq f$ validators are adversarial (so, $n-\freal = 3f + 1 - \freal$ are honest).
    Each certified transaction has $2f+1$ signatures, of which at most $\freal$ are from adversarial validators.
    All other signatories are honest, thus if $\tx \in \certtxs_v^{-}$ is certified, at least $(2f + 1 - \freal)|\tx.\diff|$ is deducted from the total budget of all honest validators (note that $\tx.\diff < 0$).
    Suppose, for contradiction, that a set of decrement transactions $\certtxs_v^{-}$ are certified such that
    $\Val{\certtxs_v^{-}} < -\frac{1}{\budFrac} \BudTotal_v$.
    Then, the average budget of all honest nodes is
    $\BudTotal \leq \BudTotal_v + \frac{2f+1 - \freal}{3f+1-\freal} \Val{\certtxs_v^{-}} \leq \BudTotal_v + \frac{f+1}{2f+1} \Val{\certtxs_v^{-}} < 0$.
    This is a contradiction because no honest validator signs a transaction that would cause its budget to fall below $0$ (\Cref{alg:bounded-counter-owned}~\Cref{loc:bc-owned-tx-check-budget}), so the average budget of honest validators cannot fall below $0$.
\end{proof}

Next, we will prove that the version updates introduced as Key Idea 2 preserve \globalsafetybc.
To do this, we first prove that the versions corresponding to certified transactions form a chain in which each version is the $\PrevVersion$ of the next (\Cref{lem:versions-form-chain}).
This follows from a quorum-intersection argument and the fact that the versions of transactions signed by any single honest validator form a chain.
This property allows us to define a linearly growing history of the bounded counter containing certified transactions from each version update request's $\PrevTxs$ (\Cref{def:bc-owned-proof-history}).
Using this, we show that the average budget of honest validators at the time they update their local $\Version$ to $v$ is $\budFrac$ times the balance after executing transactions in the counter's history, minus any budget deducted for certified transactions that were not included in the history (\Cref{lem:bc-owned-budget-value}).
Combining this with \Cref{lem:bc-owned-warmup}, we prove that the certified transactions across all versions do not spend more than the initial balance (\Cref{thm:bc-owned-boundedness}).
Since validators only execute certified transactions, this ensures global safety.

\begin{definition}
    \label{def:bc-owned-proof-parent}
    For any version $v \neq \initVersion$, define the parent version $\parent{v}$ as $v.\PrevVersion$.
\end{definition}

\begin{lemma}
    \label{lem:versions-form-chain}
    Let $\CV$ be the set of versions for which there exists at least one certified transaction. If $\CV \neq \emptyset$, then
    $\CV = \{v_0, ..., v_{|\CV|-1}\}$ such that
    for all $i=1,...,|\CV|-1$, $\parent{v_i} = v_{i-1}$.
\end{lemma}
\begin{proof}
    If no transactions have been certified, $\CV = \emptyset$.
    If at least one transaction is certified, then $\initVersion \in \CV$.
    This is because initially, honest validators sign only transactions with version $ \initVersion$ (\Cref{alg:bounded-counter-owned}~\Cref{loc:bc-owned-init-version,loc:bc-owned-tx-version}) and will not sign transactions for a different version until they receive a version update request containing certified transactions (\Cref{loc:bc-owned-req-valid}). Since a certificate requires at least one honest validator's signature, at least one certified transaction must have version $\initVersion$.

    For any honest validator $j$, if it updates $\Version$ from $v$ to $v'$, it must be such that $\parent{v'} = v$ (\Cref{alg:bounded-counter-owned}~\Cref{loc:bc-owned-req-valid}).
    Therefore, for all $v \in \CV$,
    there exists a sequence $\initVersion, ..., v$ in which each version is the parent of the next version.
    In other words, the versions in $\CV$ form a tree rooted at $\initVersion$ with parent links as edges.

    Now all that remains to show is that this tree is, in fact, a chain.
    That is, there is no $v, v' \in \CV$ such that $v \neq v'$ and $\parent{v} = \parent{v'}$.
    This follows from quorum intersection.
    If there was $v, v' \in \CV$ such that $v \neq v'$ and $\parent{v} = \parent{v'}$, then for both versions $v$ and $v'$, there is a set of $2f+1$ validators that signed transactions with that version. These two sets of $2f+1$ validators have at least $2(2f+1) - n = f+1$ validators in common (since $n = 3f+1$).
    However, since at most $f$ validators are adversarial, at least one honest validator signed both transactions with version $v$ and $v'$.
    However, this is a contradiction because once the honest validator signs a transaction for version $v$, it will never sign a transaction for version $v'$
    since there is no path $v, ..., v'$ in which each is a parent of the next one.
\end{proof}

\begin{definition}
    \label{def:bc-owned-proof-history}
    Define the history $\history{v}$ of a version as $\history{\initVersion} = \emptyset$ and $\history{v \neq \initVersion} = \history{\parent{v}} \cup v.\PrevTxs$.
\end{definition}

Let $\certtxs^i = \certtxs_{v_1} \cup ... \cup \certtxs_{v_i}$ be the set of certified transactions with versions up to $v_i$.
Let
$\exctxs{v_i} = \certtxs^i \setminus \inctxs{v_i}$ be the set of certified transactions not included in the history.
For any set of transactions $\CS$,
$\CS^{-} \subseteq \CS$ contains transactions $\tx$ such that $\tx.\diff < 0$, and $\CS^{+} = \CS \setminus \CS^{-}$.

\begin{lemma}
    \label{lem:bc-owned-budget-value}
    For any version $v_i \in \CV$,
    the average budget of all honest validators at the time they upgrade to version $v_{i}$ satisfies
    $\BudTotal_{v_{i}} \leq \budFrac (\initBal + \Val{\inctxs{v_i}} + \Val{\excnegtxs{v_i}})$.
\end{lemma}
\begin{proof}
    Throughout the execution, an honest validator i) starts with an initial budget of $\budFrac\initBal$, then ii) decreases its budget for every decrement transaction signed (\Cref{alg:bounded-counter-owned}~\Cref{loc:bc-owned-tx-decrease-budget}), iii) updates its budget for every certified transaction included in a version update request (\Cref{loc:bc-owned-req-update-budget}), and iv) reclaims its budget for every certified decrement transaction it had previously signed that is included in a version update request (\Cref{loc:bc-owned-req-regain-budget}).
    Suppose that $\freal \leq f$ validators are adversarial (so, $n-\freal = 3f + 1 - \freal$ are honest).
    Combining these four components,
    the average budget of all honest validators at the time they update to version $v_i$ is
    \begin{IEEEeqnarray*}{rCl}
        \BudTotal_{v_i} &\leq& \budFrac\initBal
        + \frac{2f+1-\freal}{3f+1-\freal} \Val{{\certtxs^i}^{-}} + \budFrac\Val{\history{v_i}}
        - \frac{2f+1-\freal}{3f+1-\freal}\Val{\history{v_i}^{-}} \\
        &=& \budFrac(\initBal + \Val{\history{v_i}}) + \frac{2f+1-\freal}{3f+1-\freal}\Val{\excnegtxs{v_i}} \\
        &\leq& \budFrac (\initBal + \Val{\inctxs{v_i}} + \Val{\excnegtxs{v_i}}).
    \end{IEEEeqnarray*}
\end{proof}

\begin{theorem}
    \label{thm:bc-owned-boundedness}
    The bounded counter protocol (validator code: \Cref{alg:bounded-counter-owned}) satisfies \globalsafetybc.
\end{theorem}
\begin{proof}
    For any given subset of honest validators, let $T$ be the set of transactions executed by some validator in this subset.
    Let $v_k$ be the latest version in $T$.
    Since validators only execute certified transactions (\Cref{alg:bounded-counter-owned}~\Cref{loc:bc-owned-cert-execute}),
    $T \subseteq \certtxs^k$.
    Moreover, validators execute transactions in $v_k.\PrevTxs$ before executing transactions with version $v_k$ (\Cref{alg:bounded-counter-owned}~\Cref{loc:bc-owned-cert-version}), so
    $T \supseteq \history{v_{k}}$.

    Recall that we partitioned $\certtxs^{k-1} = \inctxs{v_k} \cup \excnegtxs{v_k} \cup \excpostxs{v_k}$, that is, certified transactions with versions up to $v_{k-1}$ may be in the history of version $k$, and those that are not may be either increments or decrements.
    Given these constraints, it is sufficient to prove that $\initBal + \Val{T} \geq 0$ for the worst case $T = \inctxs{v_k} \cup \excnegtxs{v_k} \cup \certtxs_{v_k}^{-}$ where all decrement transactions and no increment transactions beyond $\history{v_k}$ are executed.
    \begin{align}
        \Val{T} & = \Val{\inctxs{v_k}} + \Val{\excnegtxs{v_k}} + \Val{\certtxs_{v_k}^{-}} \\
                & \geq \Val{\inctxs{v_k}} + \Val{\excnegtxs{v_k}}
        - \frac{1}{\budFrac}\BudTotal_{v_{k}} \quad \text{(\Cref{lem:bc-owned-warmup})}   \\
                & \geq -\initBal \quad \text{(\Cref{lem:bc-owned-budget-value})}
    \end{align}
\end{proof}

\section{Concurrency with Multiple Owners}
\label{sec:enhanced-programmability}

So far, we have seen designs to improve concurrency for transactions with owned objects. However, these techniques and prior consensus-less systems require that all owned objects accessed in a transaction have the same owner.
\sys enhances object programmability on the fast path using two ingredients: (i) multi-owner transactions, and (ii) collective objects.

\subsection{Multi-Owner Transactions}\label{sec:multi-owner-txs}
Sui requires all owned objects in a transaction to be `owned' by the same address~\cite{suiL}. \sys lifts this restriction: a transaction can reference owned objects from multiple owners. Validators must still ensure that each owned object referenced by a transaction is correctly authorized before signing a transaction. For example, consider an atomic swap transaction that takes object $A$ owned by Alice and object $B$ owned by Bob and exchanges their ownership.
For the transaction to be authorized we need two signatures over the full transaction (or a hash), one from an authorized signer of $A$ (i.e., Alice) and one from an authorized signer of $B$ (i.e., Bob).

However, enabling multi-owner transactions makes \sysname more susceptible to owned objects being locked through error or malicious behavior. For example, consider the previous scenario. If Alice signs $\tx$ first, Bob may refuse to sign, denying Alice access to her object. If Alice loses patience and tries to use $A$ in another transaction $\tx'$, then Bob may sign $\tx$ and race Alice's attempt to build a certificate. Now both $\tx$ and $\tx'$ contain $A$ and conflict which can lead to $A$ (and $B$) being locked.
Sui unlocks such objects after one day, at epoch change. But multi-owner transactions in \sysname make such conflicts more likely, so the latency of one day is unacceptable.
Thus, we develop the \fun protocol described in \Cref{sec:fast-unlock-owned-objects} that provides a resolution in seconds.

\subsection{Collective Objects}
\label{sec:collective-objects}

The second interesting class of objects that \sysname uses are Collective Objects. These objects can be accessed by multiple users concurrently, that is, a transaction on such an object can be authorized by any party (or a subset of parties) from a given set (the set may be infinite). Yet, they are processed on the fast path.

Applications for this include an NFT sale where users can add themselves to a collective set of users that will receive the NFT or an auction where users can add their bids as long as the auction is still running.
If there is no limit on the size of the set or the number of bids, then we can simply use an add-only-set data structure to process these transactions concurrently without any contention. However, if there is a limit, we must use a \emph{collective bounded counter}.

\subsection{Collective Bounded Counter}
\label{sec:shared-bounded-counter}

Following the owned bounded counter (\Cref{sec:bounded-counter}), our collective bounded counter allows multiple owners to send transactions and version updates. This works well as long as the concurrent transactions sent by all the owners together never exceed validators' budgets; then they all get certified, irrespective of the order in which they arrive.
However, the bounded counter may get locked and lose liveness under any of these conditions:
\begin{itemize}
    \item The owners attempt to spend more than the validators' budget in any version: In this case, honest validators may exhaust their budget by signing different transactions so that no transaction gets certified.
    \item The owners send two conflicting version update requests, i.e., neither request transitively includes transactions from the other version. In this case, honest validators may be split across the two versions, with neither group of validators willing to switch to the other version, causing neither version's transactions to get certified.
\end{itemize}
These circumstances would not occur if a single owner sending transactions follows the protocol specifications (\Cref{alg:bounded-counter-owned-user}).
However, a misconfigured owner or multiple owners who do not coordinate may cause such scenarios.

\begin{algorithm}[t]
    \caption{Collective bounded counter (validator logic)}
    \label{alg:bounded-counter-shared}
    \footnotesize
    \begin{algorithmic}[1]
        \LineComment{$\Call{InitBC}{}, \Call{ProcessTx}{}, \Call{ProcessCert}{}, \Call{ProcessVersionUpdateReq}{}$ same as \Cref{alg:bounded-counter-owned}}
        \vspace{1em}
        \LineComment{Upon receiving version merge request signed by the user}
        \Procedure{ProcessVersionMergeReq}{$\req$}
        \LineComment{Check current version is one of the versions being merged}
        \InlineRequire{$\Version \in \req.\PrevVersions$} \label{loc:bc-shared-merge-version}
        \LineComment{Update budget based on all txs in the history of $\req.\PrevVersions$, except those that have already been considered}
        \For{$\tx$ \textbf{in} $\funcCall{pendingTxsInHistory}{\req}$}
        \State $\Bud \gets \Bud + \budFrac * \tx.\diff$ \label{loc:bc-shared-merge-update-budget}
        \LineComment{Regain spent budget for txs included in $\req.\PrevTxs$}
        \If{$\tx \in \SignedTxs$ \textbf{and} $\tx.\diff < 0$} $\Bud \gets \Bud - \tx.\diff$ \EndIf \label{loc:bc-shared-merge-regain-budget}
        \EndFor
        \LineComment{Start signing txs for the updated version}
        \State $\Version \gets \req$ \label{loc:bc-shared-merge-update-version}
        \EndProcedure
    \end{algorithmic}
\end{algorithm}

In this section, we show how owners can unlock the bounded counter object, without consensus under optimistic conditions, by issuing a \emph{version merge request} (\Cref{alg:bounded-counter-shared}).
When the owners send a \emph{version merge request} that contains a set of versions $\PrevVersions$ to merge,
a validator processes the request if its current local $\Version$ is one of $\PrevVersions$ (\Cref{alg:bounded-counter-shared}~\Cref{loc:bc-shared-merge-version}).
This allows validators locked on any of the conflicting versions to adopt the merged version while also ensuring that each validator processes versions in a linear order.
Just as in version update requests, the validator updates its budget by $\budFrac$ times the value of every pending transaction in the history of the merged versions (\Cref{alg:bounded-counter-shared}~\Cref{loc:bc-shared-merge-update-budget}). In this case, the history contains all transactions included in the version update requests for $\PrevVersions$ recursively, and pending are the ones for which the budget has not been updated previously.
As in version update requests, the validator reclaims its budget for transactions in the history that it had signed (\Cref{alg:bounded-counter-shared}~\Cref{loc:bc-shared-merge-regain-budget}), and finally sets the new version to be $\req$ (\Cref{loc:bc-shared-merge-update-version}).

\tikzset{
    startstop/.style={rectangle, draw, dashed, text centered, minimum height=1.5em, minimum width=3em, font=\footnotesize},
    process/.style={circle, draw, minimum size=1.5em, text centered, font=\scriptsize},
    arrow/.style={thick,->,>=stealth},
    dashedarrow/.style={thick, dashed,->,>=stealth},
}

\begin{figure}
    \centering
    \begin{tikzpicture}[x=1cm,y=0.75cm]
        \node (v0) [startstop] at (0,0) {$\initVersion$};
        \node (node1) [process] at (0.5,-1) {-1};
        \node (node2) [process] at (1.5,-1) {-1};
        \node (node3) [process] at (2.5,-1) {-1};
        \node (node4) [process] at (-0.5,-1) {-1};
        \node (node5) [process] at (-1.5,-1) {-1};
        \node (node6) [process] at (-2.5,-1) {-1};
        \node (v1) [startstop] at (-1,-2) {$v_1$};
        \node (v2) [startstop] at (2,-2) {$v_2$};
        \node (v3) [startstop] at (0.5,-3) {$v_3$};
        \node (merge1) [left=0cm of v3, anchor=east] {\footnotesize Merge};
        \node (node7) [process] at (1,-4) {-1};
        \node (node8) [process] at (0,-4) {-1};
        \node (node9) [process] at (2,-4) {-1};
        \node (node10) [process] at (-1,-4) {-1};
        \node (excess) [left=0cm of node10, anchor=east] {\footnotesize Txs $>$ budget};
        \node (v4) [startstop] at (0.5,-5) {$v_4$};
        \node (merge2) [left=0cm of v4, anchor=east] {\footnotesize Merge};

        \draw [dashedarrow] (node1) -- (v0);
        \draw [dashedarrow] (node2) -- (v0);
        \draw [dashedarrow] (node3) -- (v0);
        \draw [dashedarrow] (node4) -- (v0);
        \draw [dashedarrow] (node5) -- (v0);
        \draw [dashedarrow] (node6) -- (v0);
        \draw [arrow] (v1) -- (node1);
        \draw [arrow] (v2) -- (node2);
        \draw [arrow] (v2) -- (node3);
        \draw [arrow] (v1) -- (node4);
        \draw [arrow] (v1) -- (node5);
        \draw [arrow] (v1) -- (node6);
        \draw [arrow] (v3) -- (v1);
        \draw [arrow] (v3) -- (v2);
        \draw [dashedarrow] (node7) -- (v3);
        \draw [dashedarrow] (node8) -- (v3);
        \draw [dashedarrow] (node9) -- (v3);
        \draw [dashedarrow] (node10) -- (v3);
        \draw [arrow] (v4) -- (v3);
    
    \end{tikzpicture}
    \caption{Version updates and merges in the collective bounded counter. Version $v_3$ results from a merge request on two conflicting version updates $v_1$ and $v_2$. Version $v_4$ results from a merge request on version $v_3$ for which the user sent transactions exceeding the budget.
    }
    \label{fig:bc-shared-versions-example}
\end{figure}
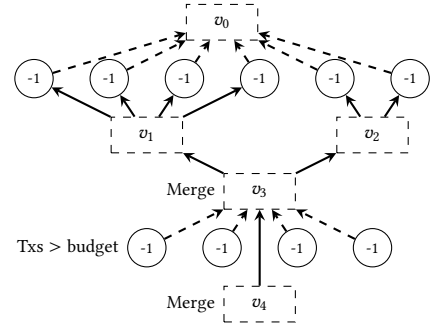

\Cref{fig:bc-shared-versions-example} shows an example of version updates and merges in the collective bounded counter.
If the owners accidentally sent two conflicting version update requests (e.g., $v_1,v_2$ in \Cref{fig:bc-shared-versions-example}),
upon sending a version merge request, validators locked on either version will switch to signing transactions for the merged version, restoring liveness.
Finally, if the owners sent transactions spending more than the budget (e.g., for $v_3$ in \Cref{fig:bc-shared-versions-example}), the owners can send a version merge request containing only one previous version.
This will cause validators to update their version without changing their budget (no new certified transactions included).
The owners can then reissue transactions, ensuring this time not to send transactions spending more than the budget.
Thus, in optimistic cases, where a single owner was misconfigured or crashed,
or the bounded counter had temporary high contention, the bounded counter can be unlocked without requiring consensus among the validators.
In cases of continuously high contention, owners must use \fun (\Cref{sec:fast-unlock}) which uses consensus to unlock their bounded counter.
In \Cref{app:bc-shared-safety-proof}, we prove the safety of the collective bounded counter.

\section{Fast Unlock Protocol} \label{sec:fast-unlock}
As discussed in \Cref{sec:enhanced-programmability}, concurrency for multi-owner transactions and collective objects does not come free. They increase the chance that transactions diverge the view of the validators, leading to a loss of liveness.
This is already a problem in Sui~\cite{sui} and Mysticeti~\cite{mysticeti} for clients that run buggy software and may issue conflicting transactions on an owned object, i.e., transactions operating on the same version of an owned object (cf. \Cref{sec:background-data-structures}). Their current solution is to wait until the end of the epoch, at which point they run an atomic snapshot sub-protocol as part of the epoch change and then drop all partial states. As a result, at the start of the new epoch, all validators have exactly the same state, and all objects can be safely accessed again.
In Sui, epoch changes occur once per day.
Since multi-owner transactions and collective objects make locks more likely, the latency of one day is unacceptable.
To remedy this issue, \sysname introduces a \fun functionality.
On a high level, \fun is a generalization of the merge functionality introduced in \Cref{sec:shared-bounded-counter}. The merge operation in \Cref{sec:shared-bounded-counter} requires an honest owner to drive it to completion and only applies to commutative transactions. For arbitrary transactions, we require consensus to decide which among conflicting transactions to accept.

\subsection{Baseline \fun Protocol} \label{sec:fast-unlock-owned-objects}

For simplicity, we show how the user can unlock a single object by executing a no-op or adopting one of the conflicting transactions.
\Cref{sec:multi-unlock} extends the basic protocol to execute a new transaction instead of a no-op. Additionally, we describe at the end of this section how validators can detect two conflicting transactions on the system and  \emph{automatically} trigger an unlock.

\para{New Persistent Data Structures}
Each \sysname validator maintains a set of persistent tables abstracted as key-value maps, with the usual $\mathsf{contains}$, $\mathsf{get}$, and $\mathsf{set}$ operations.
The table
$$\lockdb[\Okey] \rightarrow \Cert \text{ or } \none$$
maps $\Okey = (\Oid,\Version)$, an object's identifier and version, to a certificate $\cert$, or $\none$ if the object's version exists but the validator does not hold a certificate for it.
The map
$$\unlockdb[\Okey] \rightarrow \unlocked \text{, } \confirmed \text{, or } \none$$
records whether a transaction over the specified object version is involved in a current \fun instance ($\unlocked$), has been sequenced by consensus ($\confirmed$), or none of the above ($\none$).

All new owned object entries start with $\unlockdb[\Okey]$ set to \none. Once a transaction certificate is sequenced through consensus, it is always executed (whether it is for a shared object transaction or an owned-object-only transaction) and all owned object entries have $\unlockdb[\Okey]$ set to \confirmed.

\para{\fun Protocol Description} \label{sec:unlock-protocol}
To safely unlock an object, the user interactively constructs a proof, called a \textit{no-commit certificate}, that no transaction modifying that object has been committed or will be committed on the fast path. This proof consists of a message signed by a quorum of validators attesting that they have not already executed a transaction on $\Okey$, and promising that they will not execute any transaction on $\Okey$ in the fast path. After that, only certificates sequenced over consensus may affect such an $\Okey$.

\begin{figure}[t]
    \centering
    \includegraphics[width=\columnwidth]{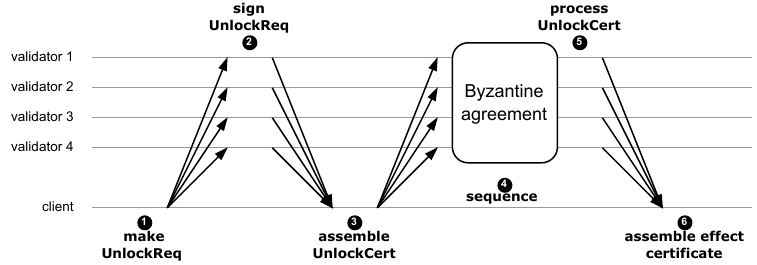}
    \caption{
        \fun interactions between a user and validators to unlock an object.
    }
    \label{fig:fast-unlock}
\end{figure}
\begin{algorithm}[t]
    \caption{Process unlock requests}
    \label{alg:process-unlock-tx}
    \footnotesize
    \begin{algorithmic}[1]
        \LineComment{Handle $\unlockreq$ messages from users.}
        \Procedure{ProcessUnlockTx}{\unlockreq}
        \LineComment{Check (\ref{alg:process-unlock-tx}.1): Check Auth. (\Cref{sec:unlock-protocol})}
        \If{$!\valid{\unlockreq}$} \label{alg:line:verify-auth} \Return error \EndIf
        \LineComment{Step (\ref{alg:process-unlock-tx}.2): Check for certificates.}
        \State $\Okey \gets \unlockreq.\Okey$
        \State $\Cert \gets \lockdb[\Okey]$ \Comment{can be $\none$} \label{alg:line:check-cert}
        \LineComment{Step (\ref{alg:process-unlock-tx}.3): Record the decision to unlock.}
        \State $\unlockvote \gets \sign{\unlockreq,\Cert}$
        \State $\unlockdb[\Okey] \gets \unlocked$ \label{alg:line:block-exec}
        \State \Return $\unlockvote$
        \EndProcedure
    \end{algorithmic}
\end{algorithm}

\Cref{fig:fast-unlock} illustrates the \fun protocol allowing a user to instruct validators to unlock a specific object.
A user first creates an \emph{unlock request} specifying the object they wish to unlock:
$$\unlockreq(\Okey, \auth)$$
This message contains the object's key $\Okey$ to unlock (accessible as $\unlockreq.\Okey$) and an authenticator $\auth$ ensuring the user is authorized to unlock $\Okey$.
The authenticator is composed of two parts: (i) a transaction that mutates the object in question (and potentially additional objects) which is signed by the object owner, and (ii) a proof that the party requesting the unlock can modify the object in question. The authenticator prevents rogue unlock requests for objects that are either not under contention (the transaction shows there exists a transaction that uses the object) or by parties not authorized to act on the objects.
The user broadcasts this $\unlockreq$ message to all validators~(\Cref{fig:fast-unlock}~\one).

Each validator handles the $\unlockreq$ as follows (\Cref{alg:process-unlock-tx}).
A validator first checks (\textbf{Check (\ref{alg:process-unlock-tx}.1)}) the validity of $\unlockreq$ by verifying the authenticator $\auth$. Specifically, $\auth$ must contain a valid transaction including $\Okey$, and a signature on the transaction by the owner of $\Okey$. Otherwise, the validator stops processing.
The validator attempts to retrieve a certificate $\Cert$ for a transaction on $\Okey$ if it exists (\textbf{Step (\ref{alg:process-unlock-tx}.2)}), or sets $\Cert$ to \none.
Then, the validator records that the object in $\unlockreq$ can only be included in transactions in the consensus path (\Cref{alg:line:block-exec}) by setting its entry in $\unlockdb[\Okey]$ to $\unlocked$ (\textbf{Step (\ref{alg:process-unlock-tx}.3)}). It finally returns a signed \emph{unlock vote} $\unlockvote$ to the user:
$$\unlockvote(\unlockreq, \mathbf{Option}(\Cert))$$

This message contains the authorized $\unlockreq$ and the certificate $\cert$ for some transaction consuming $\Okey$ that the validator executed~(\Cref{fig:fast-unlock}~\two).
If the validator has not executed any transaction on $\Okey$, then $\Cert = \none$.
$$\unlockcert(\unlockreq, \mathbf{Option}(\Cert)).$$

There are two cases in the creation of $\unlockcert$:
\begin{enumerate}
    \item At least one $\unlockvote$ carries a certificate. This scenario indicates that a correct validator has already executed a transaction, which implies that the object is not locked. However, this is not a proof of finality and subsequent steps may invalidate this execution.
    \item No $\unlockvote$ carries a certificate. This scenario is a `no-commit' proof as there are $f+1$ honest validators that will not process certificates ($\unlockdb$ holds $\unlocked$), thus no certificate will be executed in the fast path.
\end{enumerate}
The user submits this $\unlockcert$ for sequencing by the consensus engine~(\three).

\begin{algorithm}[H]
    \caption{Process unlock certificates}
    \label{alg:process-unlock-cert}
    \footnotesize
    \begin{algorithmic}[1]
        \LineComment{Handle $\unlockcert$ message from consensus.}
        \Procedure{ProcessUnlockCert}{$\unlockcert$}
        \LineComment{Check (\ref{alg:process-unlock-cert}.1): Check no transaction already processed (\Cref{sec:unlock-protocol}).}
        \If{$\unlockdb[\Okey] = \confirmed$}  \label{alg:line:ensure-first}
        \Return
        \EndIf
        \LineComment{Check (\ref{alg:process-unlock-cert}.2): Check cert validity (\Cref{sec:unlock-protocol}).}
        \If{$!\valid{\unlockcert}$} \Return error \EndIf \label{alg:line:check-unlock-cert}
        \LineComment{Execute $\cert$ or $\none$ (\ref{alg:process-unlock-cert}.3).}
        \State $\Cert \gets \unlockcert.\cert$
        \If {$\Cert \neq \none$ }
        $\tx \gets \Cert.\tx$
        \Else\xspace $\tx \gets \noop$
        \EndIf
        \State $\effectvote \gets \exec{\tx, \unlockcert}$ \label{alg:line:exec}
        \LineComment{Prevent execution overwrite.}
        \State $\unlockdb[\Okey] \gets \confirmed$ \label{alg:line:mark-confirmed}
        \State \Return $\effectvote$
        \EndProcedure
    \end{algorithmic}
\end{algorithm}

All correct validators observe a consistent sequence of  $\unlockcert$ messages output by consensus~(\four) and process them in order as follows (\Cref{alg:process-unlock-cert}).
A validator performs the following checks and if any fail, they ignore the certificate:
\begin{itemize}
    \item \textbf{Check (\ref{alg:process-unlock-cert}.1) } They ensure they did not already process another transaction to completion (i.e.\ $\unlockdb$ is not \confirmed) or a different $\unlockcert$ for the same $\Okey$.  %
    \item \textbf{Check (\ref{alg:process-unlock-cert}.2) } They check $\unlockcert$ is valid, that is, (i) it is correctly signed by a quorum of authorities, and (ii) the certificate $\cert$ it contains is valid or $\none$.
\end{itemize}
The validator then executes the transaction referenced by $\cert$ (Step \ref{alg:process-unlock-cert}.3) if one exists. Otherwise, if $\cert$ is $\none$, the validator undoes any transaction locally executed on the object%
\footnote{The $\unlockcert$ with $\cert$ being \none ensures such execution could not have been finalized; only a single layer of execution can ever be undone, and no cascading aborts can happen.}%
, then executes a no-op, that is, the object contents remain unchanged but its version number increases by one.
The validator finally marks every object key as $\confirmed$ to prevent future unlock certificates or checkpoint certificates from overwriting execution (\Cref{alg:line:mark-confirmed}) and returns an $\effectvote$ to the user~(\five).
The user assembles a quorum of $2f+1$ $\effectvote$ messages into an \emph{effect certificate} $\effectcert$ that determines finality~(\six).

\Cref{sec:gas} details the use of gas objects in the context of \fun and \Cref{sec:proofs} proves the safety and liveness of the \sysname system using \fun. The key insight is that an $\unlockcert$ forces transactions on the owned object to go through consensus. There, either a transaction certificate or an unlock certificate will be sequenced first and executed. If a transaction is finalized, an unlock certificate will always cause the execution of that transaction.

\para{Auto-Unlock}
The basic \fun scheme presumes that the request to unlock an object is authenticated by the owner(s) of the object. This ensures that only authorized parties can interfere with the completion of a transaction, but it also restricts who can initiate unlocking in case of loss of liveness. Alternatively, an `AutoUnlock' can be issued by validators if the fast-path protocol is embedded in the consensus protocol, as proposed by Mysticeti~\cite{mysticeti}. In such a protocol, the presence of conflicting transactions in the causal history of a consensus block is evidence of loss of liveness.
Upon seeing such evidence, validators can start locally processing a virtual unlock request posting the signed unlock requests as transactions in the consensus protocol and forming unlock certificates.

\section{Implementation} \label{sec:implementation}
We base our implementation of \sysname on Sui~\cite{sui} as it is, to our knowledge, the only blockchain currently supporting consensus-less transactions. Specifically, we fork the research codebase of Mysticeti~\cite{mysticeti-code}, which is a fork of the production codebase of Sui, but without irrelevant features such as Admission control, RPC endpoints, support for light clients, enforcement of correct genesis, etc.
Our implementation only modifies the block and transaction processing logic by adding the bounded counter, keeping the networking, storage, and cryptography layers untouched.
We open-source our implementation of \sysname and our orchestration tools to ensure reproducibility of our results\footnote{\codelink}.

\section{Evaluation} \label{sec:evaluation}
We evaluate the throughput and latency of \sysname through experiments conducted on Amazon Web Services (AWS), demonstrating its performance improvements over the state-of-the-art.

We compare \sysname with the consensus-less fast path of Sui~\cite{sui}, called Mysticeti-FPC~\cite{mysticeti} as to our knowledge, Sui is the only blockchain supporting consensus-less transactions.
We did not compare with other consensus-less systems, including FastPay~\cite{fastpay}, Astro~\cite{astro}, Zef~\cite{zef}, and Brick~\cite{brick} because they only support payments and are thus not adapted to showcase loads under high concurrency\footnote{Parallelizing payments issued from an account can be achieved by splitting the available balance into multiple accounts before initiating the payments.}. Furthermore, these systems lack a mechanism to unlock transactions and thus cannot optimistically handle contention.

Our evaluation demonstrates the following claims:
\begin{itemize}
    \item\textbf{C1}: Clients of \sysname submitting a commutative load experience a lower latency and higher throughput than those of the Sui baseline.
    \item\textbf{C2}: There is no noticeable performance difference between Sui using owned objects and \sysname using bounded counters for loads with no contention (i.e., parallel). In other words, there is no performance trade-off in adopting \sysname.
    \item\textbf{C3}: Operations under (crash) faults do not overly penalize \sysname in comparison to Sui. That is, both systems observe similar performance degradation when validators are faulty.
\end{itemize}
Note that evaluating the performance of BFT protocols in the presence of Byzantine faults is an open research question~\cite{twins}, and state-of-the-art evidence relies on formal proofs. %

\subsection{Experimental Setup} \label{sec:setup}

We deploy all systems on a geo-distributed network of validators, each running on a dedicated machine. \Cref{sec:detailed-setup} details the precise machine specs, validator configuration, and the network setup.

In the following graphs, each data point is the p50 latency and the error bars represent the p90 latency (error bars are sometimes too small to be visible on the graph). We instantiate several geo-distributed benchmark clients within each validator, submitting transactions at a fixed rate for 5 minutes.
We increase the load of transactions sent to the systems and record the throughput and latency. As a result, each plot illustrates the `steady state' latency of the system under low load and the maximum throughput it can serve after which latency grows steeply.
Transactions in the benchmarks contain 512 bytes. The ping latency between the validators varies from 50ms to 250ms.

By \emph{latency}, we mean the time between when the client submits the transaction and when the transaction is finalized by the validators. By \emph{throughput}, we mean the number of transactions finalized per second during the run.

\subsection{Benchmark under Commutative Load}

\Cref{fig:sequential} compares the throughput and latency experienced by clients of Sui and \sysname submitting a load of \emph{commutative} transactions. In Sui, these transactions are implemented using operations on the same owned object~\cite{sui}, whereas, in \sysname, they rely on bounded counter withdrawals with values much lower than the available balance (\Cref{sec:bounded-counter}).
Specifically, we measure the maximum rate at which a client can submit these transactions and the corresponding end-to-end latency. Both systems operate in a failure-free wide-area network (WAN) environment, configured with committees of 10 and 50 validators to reflect small and large committee setups.

\begin{figure}[t]
    \centering
    \begin{tikzpicture}[]
        \begin{axis}[
                mysimpleplot,
                xmode=log,
            ]

            \addplot [line1] table [x=throughput,y=latency, y error=yerr] {data/latex-txt/sui-sequential-10-0.txt};
            \addlegendentry{Sui - 10 nodes};
            \label{plt:experiment-sequential-sui10};

            \addplot [line2] table [x=throughput,y=latency, y error=yerr] {data/latex-txt/sui-sequential-50-0.txt};
            \addlegendentry{Sui - 50 nodes};
            \label{plt:experiment-sequential-sui50};

            \addplot [line3] table [x=throughput,y=latency, y error=yerr] {data/latex-txt/bcounter-inf-10-0.txt};
            \addlegendentry{\sysname - 10 nodes};
            \label{plt:experiment-sequential-stingray10};

            \addplot [line4] table [x=throughput,y=latency, y error=yerr] {data/latex-txt/bcounter-inf-50-0.txt};
            \addlegendentry{\sysname - 50 nodes};
            \label{plt:experiment-sequential-stingray50};
        \end{axis}
    \end{tikzpicture}%
    \caption{
        Comparing throughput and latency of Sui and \sysname with a \emph{commutative} load. WAN measurements with $10$ and $50$ validators. %
    }
    \label{fig:sequential}
    \Description{}
\end{figure}
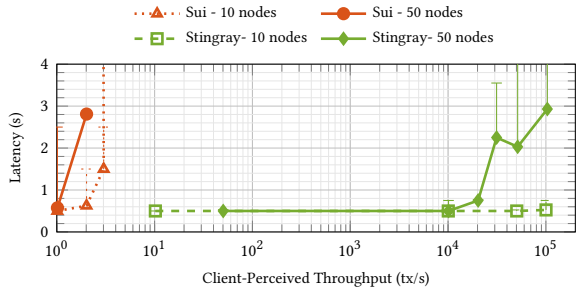

\begin{table}[t]
    \centering
    \begin{tabular}{lrrrr}
        \toprule

        Protocols \hskip 2em & 1 tx      & 10 tx   & 100 tx        \\
        \midrule
        Sui                  & 400-500ms & 4-5s    & $\approx$ 50s \\
        \sysname             & < 500ms   & < 500ms & < 500ms       \\
        \bottomrule
    \end{tabular}
    \caption{
        Average total time required to submit a load of 1, 10, and 100 commutative transactions to Sui and \sysname.
    }
    \label{tab:spend-time}
\end{table}

As expected, Sui clients can submit only about two commutative transactions per second. This is because Sui fails to exploit the commutativity of these transactions and processes them sequentially because it detects false dependencies based on memory access patterns. Consequently, a client must wait approximately 500ms for one transaction to complete before submitting the next.
Despite Sui's low baseline latency~\cite{mysticeti,sui}, this commutative load results in a latency proportional to the number of transactions submitted. Consequently, clients perceive significantly higher overall latency as the transaction count increases. For example, as shown in \Cref{tab:spend-time}, a single transaction incurs the state-of-the-art latency of 500 ms, but submitting 100 of these transactions causes the total latency to grow linearly to 50 seconds.

In contrast, \Cref{fig:sequential} demonstrates that \sysname enables parallel submission of commutative transactions, maintaining latencies under 500 ms. This performance holds even for workloads of 10-20k commutative transactions with a large committee of 50 validators (note the log scale on the x-axis) or 100k transactions with a smaller committee of 10 validators. This improvement stems from \sysname's use of the bounded counter, which processes transactions concurrently. As a result, transaction latency remains unaffected by the submission rate until the system reaches saturation.
Throughout these benchmarks, the CPU utilization of the validators of both systems remains roughly below 20\% and the validators consume less than 10GB of memory (when experiencing the highest loads).

These results validate our claim \textbf{C1}: clients of \sysname experience lower latency and higher throughput than those of the Sui baseline when handling commutative transaction loads.

\subsection{Benchmark under Parallel Load}

\Cref{fig:parallel} compares the throughput and latency experienced by clients of Sui and \sysname submitting a load of \emph{independent} transactions. Both systems operate in a failure-free wide-area network (WAN) environment, configured with committees of 10 and 50 validators.

\begin{figure}[t]
    \centering

    \begin{tikzpicture}[]
        \begin{axis}[
                mysimpleplot,
                xmode=log,
            ]

            \addplot [line1] table [x=throughput,y=latency, y error=yerr] {data/latex-txt/sui-parallel-10-0.txt};
            \addlegendentry{Sui - 10 nodes};
            \label{plt:experiment-parallel-sui10};

            \addplot [line2] table [x=throughput,y=latency, y error=yerr] {data/latex-txt/sui-parallel-50-0.txt};
            \addlegendentry{Sui - 50 nodes};
            \label{plt:experiment-parallel-sui50};

            \addplot [line3] table [x=throughput,y=latency, y error=yerr] {data/latex-txt/bcounter-inf-10-0.txt};
            \addlegendentry{\sysname - 10 nodes};
            \label{plt:experiment-parallel-stingray10};

            \addplot [line4] table [x=throughput,y=latency, y error=yerr] {data/latex-txt/bcounter-inf-50-0.txt};
            \addlegendentry{\sysname - 50 nodes};
            \label{plt:experiment-parallel-stingray50};

        \end{axis}
    \end{tikzpicture}%

    \caption{
        Comparing throughput and latency of Sui and \sysname with a \emph{parallel} load. WAN measurements with $10$ and $50$ validators. %
    }
    \label{fig:parallel}
    \Description{}
\end{figure}
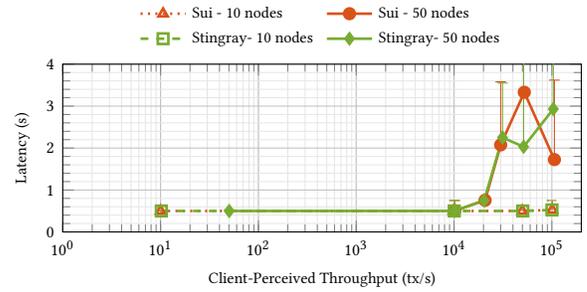

In contrast to the previous benchmark, the transactions in this benchmark are implemented using operations on different owned objects. As a result, both systems process these transactions concurrently, and the throughput and latency experienced by the clients of both systems are similar. In both cases, clients experience a latency of less than 500 ms and a throughput of 10-20k transactions per second with a large committee of 50 validators and 100k transactions per second with a small committee of 10 validators.

This result validates our claim \textbf{C2}: there is no noticeable performance difference between Sui and \sysname for loads with no contention (in this case, both systems use owned objects), i.e., there is no performance trade-off in adopting \sysname.

\subsection{Benchmark under Faults}

\Cref{fig:faults} compares the throughput and latency experienced by clients of Sui and \sysname submitting a load of independent transactions when a committee of 10 validators experiences 3 (crash) faults, which is the maximum number of faults that can be tolerated in this systems' configuration.
The results show that both systems observe similar performance degradation when validators are faulty. In both cases, the throughput drops to about 70k transactions per second, and the latency increases to about 1 second.
This result validates our claim \textbf{C3}: operations under (crash) faults do not overly penalize \sysname's clients in comparison to Sui's.

\begin{figure}[t]
    \centering
    \begin{tikzpicture}[]
        \begin{axis}[
                mysimpleplot,
                xmode=log,
            ]

            \addplot [line1] table [x=throughput,y=latency, y error=yerr] {data/latex-txt/sui-parallel-10-0.txt};
            \addlegendentry{Sui - 10 nodes};
            \label{plt:experiment-parallel-faults-sui-10-0};

            \addplot [line2] table [x=throughput,y=latency, y error=yerr] {data/latex-txt/sui-parallel-10-3.txt};
            \addlegendentry{Sui - 10 nodes (3 faults)};
            \label{plt:experiment-parallel-faults-sui-10-3};

            \addplot [line3] table [x=throughput,y=latency, y error=yerr] {data/latex-txt/bcounter-inf-10-0.txt};
            \addlegendentry{\sysname - 10 nodes};
            \label{plt:experiment-parallel-faults-stingray-10-0};

            \addplot [line4] table [x=throughput,y=latency, y error=yerr] {data/latex-txt/bcounter-inf-10-3.txt};
            \addlegendentry{\sysname - 10 nodes (3 faults)};
            \label{plt:experiment-parallel-faults-stingray-10-3};

        \end{axis}
    \end{tikzpicture}%
    \caption{
        Comparing throughput and latency of Sui and \sysname with a \emph{parallel} load. WAN measurements with $10$ validators, $3$ faults.%
    }
    \label{fig:faults}
    \Description{}
\end{figure}
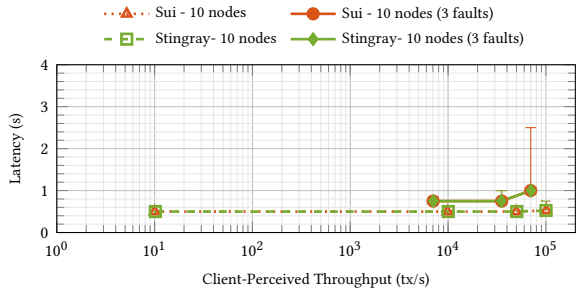

\section{Related Work}

\sysname is closely related to three research directions: consensus-less (fast path) blockchains, parallel execution engines, and replicated data types. Consensus-less blockchains were originally proposed for payments both theoretically~\cite{guerraoui19consensus,auvolat2020moneytransfer,astro} and in practice~\cite{fastpay} and achieve the lowest possible latency. However, these systems support only payments and no programmability, and do not support validators' reconfiguration.
Similarly, Groundhog~\cite{groundhog} foregoes consensus for commutative transactions but doesn't allow non-commutative ones.
The Sui Lutris system~\cite{suiL}, implemented in the Sui blockchain~\cite{sui}, combines FastPay~\cite{fastpay} and the Bullshark~\cite{bullshark} consensus protocol to deliver low-latency payments and full programmability.
Further, Mysticeti~\cite{mysticeti} addresses Sui's redundant broadcasting and high signature verification costs.
However, all these works adopt a conservative approach to the consensus-less path, limiting it to transactions involving state owned by a single account which avoids contention unless the account owner equivocates.
The generic broadcast~\cite{raykov2011commutative,pedone2001generic} framework uses consensus-less broadcast for non-conflicting messages and consensus for others, but does not specify which transactions conflict.
\sysname expands the design space by allowing transactions that are commutative or have a low risk of contention to run on the consensus-less path. In case of contention, \sysname leverages \fun to restore liveness efficiently.

Parallel execution in blockchains is a relatively new research area
led by Solana~\cite{solana} and FuelVM\footnote{https://docs.fuel.network/docs/intro/what-is-fuel/}.
On the research side, Block-STM~\cite{blockstm} focuses on shared memory and executes blocks of transactions instead of streaming. This creates a tension between high-throughput and low latency as high-throughput needs a high batch size, but collecting this batch increases the latency. On the other hand, Pilotfish~\cite{pilotfish} focuses on multi-machine execution and follows Sui's streaming architecture.
Sui also supports parallel execution in a single machine
but only for transactions accessing different memory locations.
\sysname improves upon the state of the art by enabling the parallel execution of transactions that access the same memory location, as long as they do not conflict. While this paper emphasizes parallelization on the consensus-less path, the same principles apply even more easily post-consensus as there is no risk of losing liveness due to equivocation.

Another closely related research area is replicated data types.
Unlike conflict-free data types (CRDTs)~\cite{shapiro2011conflict,keep-calm-crdt-on}, our bounded counter supports non-commutative and non-inflationary state transitions.
While some previous work~\cite{kleppmann2022making,frey2024process,capretto2022setchain,capretto2023setchain,chai2014byzantine,cholvi2021byzantinegset,nasirifard2023orderlesschain,almeida2024blocklace,zhao2015optimistic} developed Byzantine fault-tolerant CRDTs and others developed bounded counters that are safe under no faults~\cite{barbara-demarcation,oneil-escrow} and crash faults~\cite{balegas2015extending},
our work achieves the best of both through the first Byzantine fault-tolerant bounded counter.
RDTs with non-commutative operations inherently require coordination among the replicas~\cite{bailis2014coordination,kleppmann2020byzantine,bazzi2024fractional}, even without Byzantine faults.
Accordingly, our bounded counter also requires coordination but the number of rounds of coordination is at most logarithmic in the counter's initial value.

RapidLane~\cite{rapidlane} enables concurrent transactions by deferring execution and instead predicting transaction outcomes without reading the object's state. This allows optimistic concurrent processing, including bounded counters. In contrast, \sysname ensures correct execution upfront, avoiding the need for rollbacks due to incorrect predictions.
Bazzi et al.~\cite{bazzi2024fractional} enable concurrent payments using small random disjoint quorums to certify each transaction. However, their approach is probabilistic and tolerates only $1/8$ faulty validators, while \sysname is deterministic and tolerates up to $1/3$ faulty validators.

\ifpublish
  \section*{Acknowledgments}
This work is sponsored by Mysten Labs.
We thank George Danezis for initial discussions and feedback and Igor Zablotchi, David Tse, and Roger Wattenhofer for feedback.
\fi

\bibliographystyle{ACM-Reference-Format}
\bibliography{refs}

\appendix

\section{Algorithms and Proofs for the Bounded Counter}
\label{sec:app-bcounter-algs-proofs}

\subsection{Liveness Proof for the Owned Bounded Counter}
\label{sec:app-bc-owned-liveness-proof}

\begin{theorem}
    \label{thm:bc-owned-liveness}
    The bounded counter protocol (validator code: \Cref{alg:bounded-counter-owned}, user code: \Cref{alg:bounded-counter-owned-user}) satisfies liveness.
    If the user sends only decrement transactions, all transactions sent by an honest user will be executed by honest validators in $\CO(\log(\initBal))$ rounds.
\end{theorem}
\begin{proof}
    First, we show that when the user runs \Cref{alg:bounded-counter-owned-user}, the budget $\Bud$ computed by the user matches the budget computed by each honest validator.
    Let $\Bud_v^c$ be the budget computed by the user and let $\Bud_v^j$ be the budget computed by a validator $j$ at the time when they each update their local $\Version$ to $v$.
    We will first show that for all validators $j$, $\Bud_v^j = \Bud_v^c = \budFrac(\initBal + \history{v})$.

    To show this, first note that for every honest validator $j$, $\Bud_v^j = \budFrac(\initBal + \history{v})$. This can be seen from \Cref{lem:bc-owned-budget-value} along with the observation that $\excnegtxs{v} = \emptyset$, i.e., there are no certified transactions not included in the history of the latest version update, because the user includes all transactions it sent in the version update (\Cref{alg:bounded-counter-owned-user}~\Cref{loc:bc-owned-user-create-req}).
    Second, we see that $\Bud_v^c = \budFrac(\initBal + \history{v})$.
    This is because, as argued above, $\history{v}$ is exactly the set of transactions ever sent by the user, and for each sent transaction, the user's budget updates in \Cref{alg:bounded-counter-owned-user}~\Cref{loc:bc-owned-user-tx-deduct-budget,loc:bc-owned-user-req-update-budget,loc:bc-owned-user-req-regain-budget} have the net effect of updating the budget by $\budFrac * \tx.\diff$ for each transaction $\tx$ sent by the user.

    Finally, for any given version $v$, the set of decrement transactions $\CS$ the user sends satisfies $\Bud_v^c + \Val{\CS} \geq 0$ (\Cref{alg:bounded-counter-owned-user}~\Cref{loc:bc-owned-user-check-budget,loc:bc-owned-user-tx-deduct-budget}).
    Therefore, every honest validator signs these transactions eventually (once all messages are delivered).
    If at most $f$ validators are adversarial, every transaction gets certified, and thus eventually executed by honest validators.

    Within a given version, all transactions sent by the user get certified concurrently, without any additional rounds of communication.
    At each version update, validators must wait to execute transactions with the previous version before executing transactions with the new version, thus successive versions are processed sequentially.
    If the user sends only decrement transactions, at each version update, the budget gets scaled by a factor $\budFrac < 1$.
    Thus, after $\CO(\log(\initBal))$ version updates, the budget will fall below a pre-specified bound $\textsf{minBudget}$, after which the user can convert the bounded counter object to a standard owned object to spend the remaining balance in a single transaction.
\end{proof}
\section{Safety Proof for the Collective Bounded Counter}
\label{app:bc-shared-safety-proof}

In this section, we prove the safety properties of the collective bounded counter (described in \Cref{sec:shared-bounded-counter}, \Cref{alg:bounded-counter-shared}).

Eventual consistency (\Cref{thm:bc-owned-eventual-cons}) continues to hold because when an honest validator executes a transaction upon seeing a certificate, it broadcasts the certificate to all other validators.
Similarly, validity continues to hold for the collective bounded counter.

To prove \globalsafetybc, we begin by noting that \Cref{lem:bc-owned-warmup} holds for the collective bounded counter in the same way as for the owned bounded counter.
That is, within a single version, any subset of the certified transactions does not spend too much.
For clarity, \Cref{lem:bc-owned-warmup} is recapped below.
The proof is the same as for the owned bounded counter because the rules for signing a transaction within a given version are the same in the collective bounded counters.

\begin{lemma}
    \label{lem:bc-shared-warmup}
    For any version $v$, let $\BudTotal_v$ be the average budget of all honest validators at the time when they set $\Version$ to $v$.
    Let $\certtxs_v$ be the set of certified transactions with version $v$.
    Then, for all $T \subseteq \certtxs_{v}$: 
    $\Val{T} \geq -\frac{1}{\budFrac} \BudTotal_v$.
\end{lemma}

We now extend the remainder of the owned bounded counter's security proof (\Cref{sec:bounded-counter-owned-proof}) to the collective bounded counter.
Since the collective bounded counter allows merging multiple versions to create a new version, we consider the directed acyclic graph (DAG) formed by the versions, in which, unlike the owned bounded counter, each version may have multiple parents.

In the collective bounded counter, versions can be changed through a version update request (as in \Cref{alg:bounded-counter-owned}) or through a version merge request (as in \Cref{alg:bounded-counter-shared}). \Cref{def:bc-shared-proof-history} generalizes the definition of a version's parent to capture both these cases.

\begin{definition}
    \label{def:bc-shared-proof-parent}
    For any version $v \neq \initVersion$, define the set of parent versions $\parents{v}$ as $\{v.\PrevVersion\}$ if $v$ is a version update and $v.\PrevVersions$ is $v$ is a version merge.
\end{definition}

Even though the structure of the DAG formed by the versions is different, the key property of \Cref{lem:versions-form-chain} continues to hold. That is, the set of versions for which there exists at least one certified transaction forms a chain.
However, these versions may not be consecutive in the chain.

\begin{definition}
    There exists a path from $v$ to $v'$, indicated by $v \to v'$, if for some $k \geq 1$, there exists a sequence $v_1,...,v_k$ such that $v_1 = v'$, $v_k=v$, and for all $1 < i \leq k$, $v_{i-1} \in \parents{v_i}$.
    If $v \to v'$ doesn't hold, we write $v \not\to v'$.
\end{definition}

\begin{lemma}
    \label{lem:versions-form-chain-shared}
    Let $\CV$ be the set of versions for which there exists at least one certified transaction. If $\CV \neq \emptyset$, then
    $\CV = \{v_1, ..., v_{|\CV|}\}$ such that 
    $v_1 \to \initVersion$ and
    for all $i=2,...,|\CV|$, $v_i \to v_{i-1}$.
\end{lemma}
\begin{proof}
    If no transactions have been certified, $\CV = \emptyset$.
    If at least one transaction is certified, then $v_1 \to \initVersion$.
    This is because initially honest validators sign only transactions with version $ \initVersion$ (\Cref{alg:bounded-counter-owned}~\Cref{loc:bc-owned-init-version,loc:bc-owned-tx-version}) and will not sign transactions for a different version until they receive a version update request containing certified transactions (\Cref{alg:bounded-counter-owned}~\Cref{loc:bc-owned-req-valid}) or a version merge request in which $\initVersion \in \PrevVersions$ (\Cref{alg:bounded-counter-shared}~\Cref{loc:bc-shared-merge-version}). 

    For any honest validator $j$, if it updates $\Version$ from $v$ to $v'$, it must be such that $v 
    \in \parents{v'}$ (\Cref{alg:bounded-counter-owned}~\Cref{loc:bc-owned-req-valid}, \Cref{alg:bounded-counter-shared}~\Cref{loc:bc-shared-merge-version}).
    Therefore, for all $v \in \CV$,
    there exists a sequence $\initVersion, ..., v$ in which there is a path from each version to the next version.
    In other words, the versions in $\CV$ are part of
    a tree rooted at $\initVersion$ with parent links as edges.

    Now all that remains to show is that this tree is, in fact, a chain.
    That is, there is no $v, v' \in \CV$ such that 
    $v \not\to v'$ and $v' \not\to v$.
    This follows from quorum intersection.
    If there was $v, v' \in \CV$ such that 
    $v \not\to v'$ and $v' \not\to v$,
    then for both versions $v$ and $v'$, there is a set of $2f+1$ validators that signed transactions with that version. These two sets of $2f+1$ validators have at least $2(2f+1) - n = f+1$ validators in common (since $n = 3f+1$).
    However, since at most $f$ validators are adversarial, at least one honest validator signed both transactions with version $v$ and $v'$.
    However, this is a contradiction because an honest validator will not sign transactions for both versions $v$ and $v'$
    since there is no path from $v$ to $v'$ or from $v'$ to $v$.
\end{proof}

Further, we generalize the history of a version (\Cref{def:bc-owned-proof-history}) to capture version update and version merge requests.
Like the owned bounded counter, the collective bounded counter too updates validators' budgets accounting for all transactions in the history of the new version.
So, the proof of \Cref{lem:bc-owned-budget-value} carries forward similarly in \Cref{lem:bc-shared-budget-value}.

\begin{definition}
    \label{def:bc-shared-proof-history}
    Define the history $\history{v}$ of a version as $\history{\initVersion} = \emptyset$, $\history{v \neq \initVersion} = \history{\parent{v}} \cup v.\PrevTxs$ if $v$ is a version update, and $\history{v \neq \initVersion} = \bigcup_{v' \in \parents{v}} \history{v'}$.
\end{definition}

\begin{lemma}
    \label{lem:bc-shared-budget-value}
    For any version $v_i \in \CV$,
    the average budget of all honest validators at the time they upgrade to version $v_{i}$ satisfies
    $\BudTotal_{v_{i}} \leq \budFrac (\initBal + \Val{\inctxs{v_i}} + \Val{\excnegtxs{v_i}})$.
\end{lemma}
\begin{proof}
    Throughout the execution, an honest validator i) starts with an initial budget of $\budFrac\initBal$, then ii) decreases its budget for every decrement transaction signed (\Cref{alg:bounded-counter-owned}~\Cref{loc:bc-owned-tx-decrease-budget}), iii) updates its budget for every certified transaction included in a version update request (\Cref{alg:bounded-counter-owned}~\Cref{loc:bc-owned-req-update-budget}) or transitively included in a version merge request (\Cref{alg:bounded-counter-shared}~\Cref{loc:bc-shared-merge-update-budget}), and iv) reclaims its budget for every certified decrement transaction it had previously signed that is included in a version update request (\Cref{alg:bounded-counter-owned}~\Cref{loc:bc-owned-req-regain-budget}) or transitively included in a version merge request (\Cref{alg:bounded-counter-shared}~\Cref{loc:bc-shared-merge-regain-budget}).
    Suppose that $\freal \leq f$ validators are adversarial (so, $n-\freal = 3f + 1 - \freal$ are honest).
    Combining these four components,
    the average budget of all honest validators at the time they update to version $v_i$ is
    \begin{IEEEeqnarray*}{rCl}
        \BudTotal_{v_i} &\leq& \budFrac\initBal
        + \frac{2f+1-\freal}{3f+1-\freal} \Val{{\certtxs^i}^{-}} + \budFrac\Val{\history{v_i}}
        - \frac{2f+1-\freal}{3f+1-\freal}\Val{\history{v_i}^{-}} \\
        &=& \budFrac(\initBal + \Val{\history{v_i}}) + \frac{2f+1-\freal}{3f+1-\freal}\Val{\excnegtxs{v_i}} \\
        &\leq& \budFrac (\initBal + \Val{\inctxs{v_i}} + \Val{\excnegtxs{v_i}}).
    \end{IEEEeqnarray*}
\end{proof}

\begin{theorem}
    \label{thm:bc-shared-boundedness}
    The collective bounded counter protocol (validator code: \Cref{alg:bounded-counter-shared}) satisfies \globalsafetybc.
\end{theorem}
\begin{proof}
    For any given subset of honest validators, let $T$ be the set of transactions executed by some validator in this subset.
    Let $v_k$ be the latest version in $T$.
    Since validators only execute certified transactions (\Cref{alg:bounded-counter-owned}~\Cref{loc:bc-owned-cert-execute}),
    $T \subseteq \certtxs^k$.
    Moreover, validators execute transactions in $v_k.\PrevTxs$ before executing transactions with version $v_k$ (\Cref{alg:bounded-counter-owned}~\Cref{loc:bc-owned-cert-version}), so 
    $T \supseteq \history{v_{k}}$.

    Recall that we partitioned $\certtxs^{k-1} = \inctxs{v_k} \cup \excnegtxs{v_k} \cup \excpostxs{v_k}$, that is, certified transactions with versions up to $v_{k-1}$ may be in the history of version $k$, and those that are not may be either increments or decrements.
    Given these constraints, it is sufficient to prove that $\initBal + \Val{T} \geq 0$ for the worst case $T = \inctxs{v_k} \cup \excnegtxs{v_k} \cup \certtxs_{v_k}^{-}$ where all decrement transactions and no increment transactions beyond $\history{v_k}$ are executed.
    \begin{align}
        \Val{T} &= \Val{\inctxs{v_k}} + \Val{\excnegtxs{v_k}} + \Val{\certtxs_{v_k}^{-}} \\
        &\geq \Val{\inctxs{v_k}} + \Val{\excnegtxs{v_k}}
        - \frac{1}{\budFrac}\BudTotal_{v_{k}} \quad \text{(\Cref{lem:bc-shared-warmup})} \\
        &\geq -\initBal \quad \text{(\Cref{lem:bc-shared-budget-value})}
    \end{align}
\end{proof}
\section{Security Arguments for \fun} \label{sec:proofs}

We argue about the safety and liveness of \fun.
Intuitively, \fun does not invalidate the finality guarantees of the normal fast path operations. That is, a client holding an effect certificate can be assured that its transaction will never be reverted.

\begin{theorem}\label{th:client-safety}
    If there exists an effect certificate $\effectcert$ over a transaction $\tx$, the execution of $\tx$ is never reverted.
\end{theorem}
\begin{proof}
    We assume that the execution of \tx is reverted and show a contradiction. 
    The transaction can only be reverted if there exists an $\unlockcert$ carrying an empty certificate over an $\Okey$  modified by $\tx$.
    From Check (\ref{alg:process-unlock-tx}.2) of \Cref{alg:process-unlock-tx} a correct validator only signs an $\unlockvote$ with an empty $\cert$ only if it has not executed anything for $\Okey$. From our assumption that $\Okey$ did admit a no-op there should be $f+1$ honest validators that did not partake in the generation of the $\effectcert$ of $\tx$ and hence passed the check. Additionally, for the $\effectcert$ to exist by definition it has $2f+1$ signatories over the $\Okey$ in question, at least $f+1$ of them being honest.  This implies a total of at least $f+1 + f+1 + f = 3f+2 > 3f+1$ validators, hence a contradiction.
\end{proof}

The converse also applies, that is, if an $\unlockcert$ exists, then no $\effectcert$ over the $\Okey$ will be generated in the fast path. The proof works analogously by adding an extra check during $\effectcert$ generation in which correct validators refuse to process certificates when they recorded $\unlocked$ in their $\unlockdb$.

Next, we show that validators that might process on the consensus path both a $\cert$ (through checkpointing) and $\unlockcert$ will arrive at the same execution result. We prove the case where an $\unlockcert$ is ordered first. For this, we need to enhance the protocol of checkpointing in Sui to check the value of $\unlockdb[\Okey]$ and ignore a $\cert$ that tries to process a $\confirmed$ $\Okey$, which is a straightforward change.

\begin{theorem}
\label{thm:unlock-conflicting-cert}
    If a correct validator executes an $\unlockcert$ certificate over $\Okey$ as sequenced by the SMR engine,
    no correct validator will subsequently execute a conflicting $\cert$ as sequenced by the SMR engine.
\end{theorem}
\begin{proof}
    The proof directly follows from the safety property of the SMR engine that all validators will process certificates in the same order.
    Hence, upon processing $\unlockcert$, all honest validators mark the execution of $\Okey$ as confirmed by setting $\unlockdb[\Okey] \gets \confirmed$ (\Cref{alg:line:mark-confirmed} of \Cref{alg:process-unlock-cert}). Then, Check (\ref{alg:process-unlock-cert}.1) of \Cref{alg:process-unlock-cert} (and its dual added at the checkpoint algorithm) ensures that if any further $\cert$ or $\unlockcert$ with a conflict is given as input to the execution engine it is rejected.
\end{proof}

The converse can be proven in the same manner since we enhance the execution of $\cert$ during the checkpoint process with updating $\unlockdb[\Okey] \gets \confirmed$ after processing. Then all $\unlockcert$ on the $\Okey$ will be rejected at the Check (\ref{alg:process-unlock-cert}.1) of \Cref{alg:process-unlock-cert}.

Based on the above theorems, we can prove safety of the overall \sysname system that uses \fun.

\begin{theorem}
\label{thm:full-system-safety}
\sysname satisfies safety (\Cref{def:safety}).
\end{theorem}
\begin{proof}
    Validity holds because validators execute only transactions with valid certificates.

    Next, we prove global safety.
    For any given object $O$, let $T_p^{O}(t)$ be the set of transactions executed on that object by validator $p$ up to time $t$.
    Due to \Cref{thm:unlock-conflicting-cert}, for any two validators $p,q$, $T_p^{O}(t) \subseteq T_q^{O}(t)$ or $T_q^{O}(t) \subseteq T_p^{O}(t)$.
    Thus, $\bigcup_{p\text{ honest}} T_p^{O}(t)$ is the set executed by some honest validator. Consider the sequence $T^{O}$ made by arranging  transactions in this set in the order of the object version. Then, $T^O$ contains all transactions in $\bigcup_{p \text{ honest}} T_p^{O}(t)$ and this sequence respects the application's validity constraint since an honest validator executed transactions in this sequence.
    Finally, let $T$ be the merged sequence of $T^{O}$ for all objects, where the merge preserves the partial order for each object.
    Due to the independence of different objects, $T$ also satisfies the validity predicate, thus proving global safety.
\end{proof}

\para{Liveness argument}
Intuitively, we argue that \fun---and its composition with normal fast path operations---neither deadlocks nor enables unjustified aborts (which could starve an object from progress).

\begin{lemma}[Unlock Certificate Availability] \label{th:certificate-creation}
    A correct user can obtain an unlock certificate $\unlockcert$ over a valid $\Okey$.
\end{lemma}
\begin{proof}
    A correct validator always signs $\unlockvote$ if it passes the check of \Cref{alg:process-unlock-tx}.
    Well-formed $\unlockreq$ always come with a valid authentication path (Check (\ref{alg:process-unlock-tx}.1)), and Check (\ref{alg:process-unlock-tx}.2) always returns an $\unlockvote$.
    As a result, if $\unlockreq$ is disseminated to $2f+1$ correct validators by a correct user, they will eventually all return an $\unlockvote$. The user then aggregates those votes into a unlock certificate $\unlockcert$ over $\Okey$.
\end{proof}

\begin{theorem}[\fun Liveness] \label{th:simple-unlock-liveness}
    If a correct and authorized user initiates a fast-unlock protocol, the \Okey in question will eventually admit a new transaction.
\end{theorem}

\begin{proof}
    A correct and authorized user will eventually generate an unlock certificate by \Cref{th:certificate-creation}. Additionally from the liveness property of SMR the unlock certificate will either eventually be added as part of the SMR output or the epoch will end. If the first happens by agreement of consensus the \unlockcert will be executed by all validators, leading to the termination of the fast-unlock protocol and an updated \Okey. If the epoch ends, all locks are dropped and liveness of all \Okey are automatically available for processing.
\end{proof}
\Cref{th:simple-unlock-liveness} is sufficient for correct users as either they will manage to no-op an incorrect invocation of \Okey, drive the tranasction of a correct $\tx$ to completion, or the epoch end will automatically unblock them. This means that there will always be an available $\Okey$ to be modified.

Now that we proved that an authorized user will succeed into unblocking the $\Okey$ we also need to show that an unauthorized user will not succeed into starving legitimate users from progress through abusing fast-unlock.

\begin{theorem}[Starvation Freedom] \label{th:starvation-liveness}
    No user can successfully initiate a fast-unlock on an $\Okey$ it cannot produce an $\auth$ for.
\end{theorem}

\begin{proof}
    All honest validators check the authorization vector \auth of the requesting user (\Cref{alg:line:verify-auth} in \Cref{alg:process-unlock-tx}). This means that no honest party will lock an object without an authorization, including slow parties that have not yet seen the $\Okey$ which will reject or cache the request for later processing. As a result, by the model, there will never be sufficient $\unlockvote$ to generate an $\unlockcert$ driven by an unauthorized user.
\end{proof}

\begin{theorem}
\label{full-system-liveness}
\sysname satisfies liveness (\Cref{def:liveness}).
\end{theorem}
\begin{proof}
    First we prove progress (\Cref{def:liveness}).
    Every transaction with a valid certificate will be eventually executed unless there is an $\unlockcert$ containing $\cert=\none$.
    Moreover, if the owners of the transaction's input objects do not equivocate, there will be no $\unlockreq$ for those object (recall that only the object owners can issue an $\unlockreq$).
    This ensures progress.

    Next, we prove eventual consistency (\Cref{def:liveness}).
    If a validator $p_1$ executes a transaction, it must have seen the transaction finalized, i.e., $2f+1$ validators signed a certificate for that transaction. Therefore, at least $f+1$ honest validators must have seen a certificate $\Cert$ for that transaction.
    If there is no $\unlockcert$ for the transaction's input objects, then eventually, all honest validators will receive $2f+1$ signatures on the certificate, and thereafter execute the transaction. 
    If there is an $\unlockcert$ for one of the transaction's input objects, then $\unlockcert$ must contain $\Cert$ because at least one honest validator whose signature is in $\unlockcert$ must have seen $\Cert$.
    Therefore, even in this case, all honest validators will eventually execute the transaction.
\end{proof}

\para{Generalization to multi-object unlock}
The multi-object unlock protocol can be seen as a composition of many single-object unlock protocols (one per object) as well as a single commit protocol (for the accompanied transaction). As a result, the safety of the protocol follows from the fact that objects are independent of each other so if at least one has a prior certificate then the commit flow will lead to committing that prior certificate (which iteratively applies to all objects with prior certificates). If on the other hand, no object has a prior certificate then the workflow is the combination of the simple \fun per object together with the shared-object path of committing the transactions of Sui which is safe as proven in the original Sui paper~\cite{suiL}.
Second, we explore liveness. There are two cases: (1) all objects can be unlocked, (2) one or more objects are already certified.
The first case is exactly the same as the simple protocol of \Cref{sec:fast-unlock} and a proof would follow exactly the same structure. For the second case, we first look into the base case of a single object that is already certified which is already proven in the previous sections. For more than one objects we can see that since the validator adds all certificates in their reply and then processes each certificate separately when handling the unlock cert then there is no interaction between certificate processing and can be considered a batch of independent requests.

Finally, for liveness the accompanied transaction might need to acquire locks. This is also an independent invocation of the Sui fast-path. As a result if the transaction is valid it will either succeed or  blocks. In the latter case, the user will have to invoke fast-unlock again including in the set of to-unlock objects the newly blocked objects of the transaction. Given that there is a finite number of objects a user holds an unlock request will eventually succeed. 
\section{Contention Mitigation} \label{sec:multi-unlock}
The basic \fun protocol speeds up recovery from loss of liveness due to mistakes. However, \sys aims to support workloads on the fast path that are truly under contention. In this case, the basic protocol in \Cref{sec:fast-unlock} is insufficient, since it can result in multiple rounds of locking and no-op unlocking without any user transaction being committed. We present a protocol that proposes a new transaction during the unlock phase that is executed once the unlock is sequenced, ensuring liveness.

In the following protocol, we additionally allow users to unlock multiple objects at once.
The multi-object unlock protocol follows the same general flow as the single-object unlock protocol described in \Cref{sec:fast-unlock}. We now describe steps \one-\six depicted in \Cref{fig:fast-unlock} for the multi-unlock protocol.

\begin{algorithm}[t]
    \caption{Process unlock requests (multi-object)}
    \label{alg:process-unlock-tx-multi}
    \footnotesize

    \begin{algorithmic}[1]
        \LineComment{Handle $\unlockreq$ messages from clients.}
        \Procedure{ProcessUnlockTx}{$\unlockreq$}
        \LineComment{Check (\ref{alg:process-unlock-tx-multi}.1): Check authenticator.}
        \If{$!\valid{\unlockreq}$} \Return error \EndIf \label{alg:line:check-abort-multi}

        \LineComment{Collect certificates.}
        \State $c \gets \none$
        \For{$\Okey \in \unlockreq.\Okeys$}
        \State $c \gets c \cup \lockdb[\Okey]$ \label{alg:line:check-cert-multi}
        \EndFor
        \State $\unlockvote \gets \sign{\unlockreq, c}$

        \LineComment{Record the decision to unlock.}
        \If {$c==\none$}
        \For{$\Okey \in \unlockreq.\Okeys$}
        \State $\unlockdb[\Okey] \gets \unlocked$ \label{alg:line:block-exec-multi}
        \EndFor
        \EndIf
        \State \Return $\unlockvote$
        \EndProcedure
    \end{algorithmic}
\end{algorithm}

\para{Protocol description}
The user first creates an \emph{unlock request} specifying a set of objects to unlock:
$$\unlockreq([\Okey], \tx, \auth)$$
This message contains a list of the object's keys $[\Okey]$ to unlock (accessible as $\unlockreq.\Okeys$), a new transaction $\tx$ to execute if the unlock process succeeds, and an authenticator $\auth$ ensuring the sender is authorized to access all objects in $[\Okey]$. The user broadcasts this message to all validators~(\one).

\Cref{alg:process-unlock-tx-multi} describes how each validator handles this unlock request $\unlockreq$. They first perform Check (\ref{alg:process-unlock-tx-multi}.1) \Cref{alg:line:check-abort-multi} to check the authenticator $\auth$ is valid with respect to all objects. This check ensures that the user is authorized to mutate all the objects referenced by $\unlockreq$ and to lock all owned object referenced by $\tx$.
The validator then collects any certificates for the objects referenced by $\unlockreq$ (\Cref{alg:line:check-cert-multi}) and adds them to the response as \Cert. The validator then marks the object in $\unlockreq$ as reserved for transaction execution through consensus only (\Cref{alg:line:block-exec-multi}).

The validator finally returns an \emph{unlock vote} $\unlockvote$ to the user:
$$\unlockvote(\unlockreq, [\mathbf{Option}(\Cert)])$$
This message contains the unlock message $\unlockreq$ itself and possibly a set of certificates $[\cert]$ on transactions including the object keys referenced by $\unlockreq$ (possible empty)~(\two). If \Cert is not empty the certified transactions may have been finalized, and should be executed instead of the new transaction.

The user collects a quorum of $2f+1$ $\unlockvote$ over the same $\unlockreq$ message and assembles them into an \emph{unlock certificate} $\unlockcert$:
$$\unlockcert(\unlockreq, \Cert)$$
where $\unlockreq$ is the user-created certified unlock message and $U \Cert$ is the unions of all set of certificates received in $\unlockreq$ responses. The user submits this message to the consensus engine~(\three)
The consensus engine sequences all $\unlockcert$ messages; all correct validators observe the same output sequence~(\four).

\Cref{alg:process-unlock-cert-multi} describes how validators process these $\unlockcert$ messages after they are sequenced by the consensus engine.
The validator first ensures they did not already process another $\unlockcert$ or $\cert$ through checkpoint for the same objects keys (\Cref{alg:line:ensure-first-multi}).
They then check $\unlockcert$ is valid, that is, the validator ensures (i) it is correctly signed by a quorum of authorities, and (ii) that all certificates $[\cert]$ it contains are valid (\Cref{alg:line:check-abort-multi-cert}).
The validator can only execute the transaction $\tx$ specified by the user if $\unlockcert.\cert$ is empty (\Cref{alg:line:no-overwrite-multi}).
The validator then marks every object key of $[\Okey]$ as $\confirmed$ to prevent any future unlock requests on the $\Okey$ from overwriting execution with a different transaction (\Cref{alg:line:mark-confirmed-multi})
and returns a set of $\effectvote$ to the user~(\five).

The user assembles an $\effectvote$ from a quorum of $2f+1$ validators into an \emph{effect certificate} $\effectcert$ that determines finality~(\six).

\begin{algorithm}[t]
    \caption{Process unlock certificates (multi-object)}
    \label{alg:process-unlock-cert-multi}
    \footnotesize
    \begin{algorithmic}[1]
        \LineComment{Handle $\unlockcert$ messages from consensus.}
        \Procedure{ProcessUnlockCert}{$\unlockcert$}
        \LineComment{Check (\ref{alg:process-unlock-cert-multi}.1): Check no transaction already processed.}
        \For{$\Okey \in \unlockcert.\Okeys$}
        \If{$\unlockdb[\Okey] = \confirmed$}  \label{alg:line:ensure-first-multi}
        \State \Return
        \EndIf
        \EndFor

        \LineComment{Check (\ref{alg:process-unlock-cert-multi}.2): Check message validity.}
        \If{$!\valid{\unlockcert}$} \Return error \EndIf \label{alg:line:check-abort-multi-cert}

        \LineComment{Check (\ref{alg:process-unlock-cert-multi}.3): Can we execute the tx?}
        \State $v \gets [\;]$
        \If{\unlockcert.\cert = [\;]}  \label{alg:line:no-overwrite-multi}
        \State $\tx \gets \unlockcert.\unlockreq.\tx$
        \State $\effectvote \gets \exec{\tx, \unlockcert}$ \label{alg:line:exec-tx-multi}
        \State $v \gets \effectvote$
        \For{$\Okey \in \unlockcert.\Okeys$}
        \State $\unlockdb[\Okey] = \confirmed$
        \EndFor
        \Else
        \For{$\Cert \in \unlockcert.\cert$}
        \State $\effectvote \gets \exec{\Cert}$ \label{alg:line:exec-multi}
        \State $v \gets v \cup \effectvote$
        \For{$\Okey \in \cert.\Okeys$}
        \State $\unlockdb[\Okey] = \confirmed$ \label{alg:line:mark-confirmed-multi} %
        \EndFor
        \EndFor
        \EndIf

        \State \Return $v$
        \EndProcedure
    \end{algorithmic}
\end{algorithm}

\subsection{Handling Gas Objects} \label{sec:gas}
Typical transactions not only mutate objects but also consume a gas object to pay for the computation. If, however, the transaction is equivocated then this gas is locked as well. For this reason \sysname requires a fresh gas-object in order for consensus to process the unlock request.
Specifically together with~\Cref{alg:process-unlock-tx}, the parties should provide a fresh gas object for their request. This gas object is checked for validity along with the check in~\Cref{alg:line:verify-auth} and locked for the unlock transaction in~\Cref{alg:line:block-exec}. When the user collects the no-commit proof in the second step of the protocol, the $2f+1$ collected signatures also serve as a certificate for the gas object. The consensus then checks the validity of the certificate and spends it locally before entering~\Cref{alg:process-unlock-cert}.
Then when consensus executes the transaction, one of three scenarios may happen:

\begin{itemize}
    \item The unlock request is valid and includes a certificate. Then the execution happens as usual and both the gas object for the unlock and the gas object for the execution are consumed.
    \item The unlock request is valid and comes with a no-op. Then the gas object for unlock is consumed. If there was some locked transaction racing the \fun then the accompanying gas object is potentially blocked. The user can then explicitly unlock that gas object by running \fun.
    \item The unlock request is not processed because a checkpoint certificate already executed a transaction. Then the gas object is still consumed without altering the state of the $\Okey$.
\end{itemize}

Note that if gas objects are implemented using bounded counters, the same gas object can be spent concurrently for the transaction and the unlock, thus the above problem wouldn't exist.

\section{Detailed Experimental Setup} \label{sec:detailed-setup}
This section complements \Cref{sec:setup} by specifying the network setup and machine specs used in the benchmarks presented in \Cref{sec:evaluation}.

We deploy all systems on AWS, using \texttt{m5d.8xlarge} instances across $13$ different AWS regions:
N. Virginia (us-east-1), Oregon (us-west-2), Canada (ca-central-1), Frankfurt (eu-central-1), Ireland (eu-west-1), London (eu-west-2), Paris (eu-west-3), Stockholm (eu-north-1), Mumbai (ap-south-1), Singapore (ap-southeast-1), Sydney (ap-southeast-2), Tokyo (ap-northeast-1), and Seoul (ap-northeast-2).
Validators are uniformly distributed across those regions. Each machine provides $10$\,Gbps of bandwidth, $32$ virtual CPUs (16 physical cores) on a $3.1$\,GHz Intel Xeon Skylake 8175M, $128$\,GB memory, and runs Linux Ubuntu server $22.04$. We select these machines because they provide decent performance, are in the price range of ``commodity servers'', and satisfy the minimum required to run a Sui node as recommended by the Sui Foundation\footnote{\url{https://docs.sui.io/guides/operator/validator-config}}.

\end{document}